\documentclass[12pt]{article}

\usepackage{amssymb}
\usepackage{times}

\newtheorem{satz}{Theorem}[section]
\newtheorem{definition}[satz]{Definition}
\newtheorem{lemma}[satz]{Lemma}
\newtheorem{koro}[satz]{Corollary}

\newenvironment{proof}{\par\noindent {\it Proof:} \hspace{7pt}}%
{\hfill\hbox{\vrule width 7pt depth 0pt height 7pt}
\par\vspace{10pt}}

\newcommand{\bC}{{\mathbb C}}
\newcommand{\bD}{{\mathbb D}}
\newcommand{\bJ}{{\mathbb J}}
\newcommand{\bM}{{\mathbb M}}
\newcommand{\bN}{{\mathbb N}}
\newcommand{\bR}{{\mathbb R}}
\newcommand{\bT}{{\mathbb T}}
\newcommand{\bX}{{\mathbb X}}
\newcommand{\bZ}{{\mathbb Z}}

\newcommand{\cB}{{\cal B}}

\newcommand{\cG}{{\cal G}}
\newcommand{\cH}{{\cal H}}
\newcommand{\cJ}{{\cal J}}
\newcommand{\cI}{{\cal I}}
\newcommand{\cP}{{\cal P}}

\newcommand{\veps}{\varepsilon}

\newcommand{\GC}{Gram constant}
\newcommand{\DB}{determinant bound}

\newcommand{\dB}{\delta}
\newcommand{\BW}{{\bigwedge }}
\newcommand{\norm}[1]{{\left\Vert #1 \right\Vert}}
\newcommand{\abs}[1]{{\left\vert #1 \right\vert}}
\newcommand{\tnorm}[2]{\vert\!\vert\!\vert \, #1 \, \vert\!\vert\!\vert_{#2}}
\newcommand{\Sp}[1]{{\rm #1}}

\newcommand{\E}{{\rm e}}
\newcommand{\I}{{\rm i}}
\newcommand{\dd}{{\rm d}}

\newcommand{\CC}{\bC}
\newcommand{\DD}{\bD}
\newcommand{\hC}{\hat C}
\newcommand{\hD}{\hat D}
\newcommand{\Seq}{\bJ}
\newcommand{\mumom}{\mu}
\newcommand{\mumoma}{\mu^{a}}
\newcommand{\Xspace}{X}
\newcommand{\ball}{B_1^{(n)}}

\newcommand{\Comhg}{C_\Omega^{(\scalefunc,>)}}
\newcommand{\Comhl}{C_\Omega^{(\scalefunc,<)}}
\newcommand{\Aomhg}{A_\Omega^{(\scalefunc,>)}}
\newcommand{\chil}{\chi_<}
\newcommand{\chig}{\chi_>}
\newcommand{\Emin}{E_{\rm min}}
\newcommand{\const}{\;{\rm const}\;}
\newcommand{\abf}{\alpha}
\newcommand{\sfrac}[2]{{\textstyle{\frac{#1}{#2}}}}

\newcommand{\EKLF}{$k$--form} 

\newcommand{\Matf}{\bM_F}

\newcommand{\MI}{{\bf b}}
\newcommand{\ami}{b}
\newcommand{\scalefunc}{{\bf h}}
\newcommand{\palefunc}{{\bf g}}
\newcommand{\scf}{f}
\newcommand{\scg}{g}

\newcommand{\tK}{K}
\newcommand{\Kd}{K_d}
\newcommand{\tKd}{\tilde K_d}
\newcommand{\Kchil}{\kappa}
\newcommand{\horror}{K'}

\newcommand{\barm}{n} 
\newcommand{\ps}{\psi}
\newcommand{\psq}{\bar\psi}
\newcommand{\ul}[1]{\underline{#1}}

\begin{document}

\title{Determinant Bounds and the Matsubara 
UV Problem of Many-Fermion Systems}

\author{Walter A. de S. Pedra%
\footnote{present address: Institut f\" ur Mathematik, Universit\" at Mainz}, 
Manfred Salmhofer\\
\small
Theoretical Physics, University of Leipzig, Postfach 100920, 04009 Leipzig, 
and
\\
\small
Max--Planck Institute for Mathematics in the Sciences
\\
\small
Inselstr.\  22, 04103 Leipzig, Germany
}
\date{July 4, 2007}
 \maketitle

\begin{center}
{\sl Dedicated to J\" urg Fr\" ohlich in celebration of his 61$^{\rm st}$ birthday}
\end{center}

\begin{abstract}
\noindent
It is known that perturbation theory converges in fermionic field theory 
at weak coupling
if the interaction and the covariance are summable and 
if certain determinants arising in the expansion can be bounded
efficiently, e.g.\ 
if the covariance admits a Gram representation with a finite \GC.
The covariances of the standard many--fermion systems 
do not fall into this class due to the slow decay of the covariance
at large Matsubara frequency,  giving rise to a UV problem in 
the integration over degrees of freedom with Matsubara frequencies
larger than some $\Omega$ (usually the first step in a multiscale analysis).
We show that these covariances do not  have Gram representations
on any separable Hilbert space. 
We then prove a general bound for determinants 
associated  to chronological products
which is stronger than the usual Gram bound
and which applies to the many--fermion case. 
This allows us to prove convergence of the first integration step 
in a rather easy way, for a short--range interaction
which can be arbitrarily strong, provided $\Omega$ is chosen large enough.
Moreover,  we give -- for the first time -- nonperturbative bounds on all scales
for the case of scale decompositions of the propagator which do not impose 
cutoffs on the Matsubara frequency. 
\end{abstract}

\section{Gram representations and determinant bounds}

\noindent
Let $\bX$ be a set and  $M: \bX^2 \to \bC$, $(x,y) \mapsto M(x,y)$.
We call $M$  an $(\bX \times \bX)$-matrix and use the notation 
$M = (M_{xy})_{x,y\in\bX}$ (if $\bX = \{ 1, \ldots, n\}$, we call it as 
usual an $(n \times n)$--matrix).

\begin{definition}
\label{Gram}
Let $M$ be an $(\bX \times \bX)$-matrix.
A triple $(\cH,v,w)$, where $\cH$ is a Hilbert space and $v$ and $w$
are maps from $\bX $ to $\cH$, is called a \emph{Gram representation} 
of $M$ if 
\begin{equation}
\forall \, x,x' \in \bX: \quad M_{xx'} = \langle v_x, w_{x'}\rangle
\end{equation}
and if there is a finite constant $\gamma_M > 0$ such that
\begin{equation}
\sup\limits_{x \in \bX} \;
\max\{\norm{v_x}, \norm{w_x} \}
\le
\gamma_M .
\end{equation}
$\gamma_M$ is called the Gram constant of $M$ associated to 
the Gram representation $(\cH,v,w)$.
\end{definition}
If $M$ has a Gram representation, then the Gram estimate
(see, e.g., Lemma B.30 of \cite{msbook}) implies that 
for all $n \in \bN$ and all $x_1 , \ldots, x_n, y_1, \ldots, y_n \in \bX$, 
\begin{equation}
\abs{
\det \left[
(M_{x_ky_l})_{k,l}
\right]}
\le 
\prod\limits_{k=1}^n 
\norm{v_{x_k}} \; \norm{w_{y_k}}
\le 
{\gamma_M}^{2n} .
\end{equation}
Every $(n\times n)$--matrix $A$ has a Gram representation -- 
the equation$ A = 1 \cdot A$ (where $1$ denotes the unit matrix)
means that $A_{kl} = \langle e_k , a_l\rangle_{\bC^n}$
where $e_k$ is the $k^{\rm th}$ row of $1$ and $a_l$ is the 
$l^{\rm th}$ column of $A$. The associated Gram estimate
$|\det A | \le \prod_l \norm{a_l}_2$, the Hadamard bound, 
has associated Gram constant $\gamma_{\rm Had} = \max_{l} \norm{a_l}_2$.
Although considering diagonal matrices shows that the Hadamard
bound is optimal, the way it was derived here is basis--dependent, 
and its application in an arbitrary basis can lead to a significant 
overestimate of the determinant.
For instance, the matrix $P=v \otimes v$, where $v=(1,\ldots,1)^T \in \bC^n$
has $P_{kl} =1$ for all $k,l$, so the above Gram representation gives
$\gamma_{\rm Had} = \sqrt{n}$, thus the bound $|\det P | \le n^{n/2}$.
On the other hand, $P$ has the Gram representation $P_{ij} = w \cdot w^T$ 
with $w = n^{-1/2} (1,\ldots, 1)$, which gives the bound $\det P \le 1$. 

Thus the main issue about Gram bounds for a given class of matrices 
is not their existence of some bound, but its size, and its dependence on $n$.
Specifically, what is really needed in the  proof of 
convergence of fermionic perturbation theory given in \cite{Salm99} 
are bounds of the following type:  there is a finite constant $\delta$ such that
for all $n \in \bN$ and all
$x_1, \dots, x_n, y_1, \dots, y_n \in \bX $ 
\begin{equation}\label{eq:detbou}
 \sup\limits_{P \in \cP_{n,1}} 
\abs{\det (C_{x_i y_j} \, P_{ij})_{i,j}}
\le {\delta}^{2n}.
\end{equation}
Here $\cP_{n,1}$ denotes the set of 
complex hermitian $(n\times n)$--matrices $P=(P_{ij})$
that are nonnegative, i.e. $\sum_{i,j=1}^n P_{ij} \,
\bar{c_i} c_j \geq 0$ for all $c_1, \dots, c_n \in \bC$, 
and that have diagonal elements $P_{ii} \le 1$.
Such matrices $P$ arise naturally in interpolation constructions
of the tree expansion for the connected functions; they are positive 
if the tree expansion is chosen well
\cite{AbdessRiv98,Salm99}.

We briefly recall Lemmas 7 and 8 of  \cite{Salm99}:
The positivity of $P$ implies that $P = Q^2=Q^*Q$ with $Q \ge 0$, 
i.e.\ 
\begin{equation}\label{eq:willi}
P_{ij} = \langle q_i , q_j \rangle
\end{equation}
where $q_i$ is the $i^{\rm th}$ column of $Q$. 
Because $\langle q_i , q_i \rangle = P_{ii} \le 1$ the Gram constant of 
$P$ is bounded by $1$. 
If $C$ has a Gram representation $(\cH,v,w)$, then 
the matrix with elements 
$M_{ij} = C_{x_i y_j} \, P_{ij}$ has a Gram representation 
\begin{equation}\label{eq:willie}
M_{ij} 
=
\langle v_{x_i} \otimes q_i \; , \; w_{y_j} \otimes q_j \rangle .
\end{equation}
and $M$ has the same Gram constant as $C$ because $\norm{q_i} \le 1$ for all $i$.

\begin{definition}
\label{det.b}
Let $C$ be an $(\bX \times \bX)$-matrix.
A finite constant $\dB_C > 0$ 
is called a \emph{\DB} of $C$  
if for all $n \in \bN$  and all
$x_1, \dots, x_n, y_1, \dots, y_n \in \bX $ 
\begin{equation}\label{eq:detbound}
 \sup\limits_{p_1 \ldots, p_n, q_1, \ldots, q_n \in \ball}
\abs{\det (\langle p_i \, , \, q_j \rangle C_{x_i y_j})_{i,j}}
\le {\dB_C}^{2n}.
\end{equation}
Here $\ball = \{ \xi \in \bC^n : \norm{\xi}_2 \le 1\}$
denotes the closed $n$--dimensional unit ball. 
\end{definition}

We have replaced the supremum over $P \in \cP_{n,1}$ 
by that over a larger set in Definition \ref{det.b}
because this makes the definition robust under 
the operation of taking arbitrary 
submatrices (positivity is spoiled by that operation). 

If $C$ has a Gram representation with \GC\ $\gamma_C$, 
then $C$ also has a \DB\ $\dB_C = \gamma_C$
by the same argument as above, i.e.\ writing 
$\langle p_i \, , \, q_j \rangle C_{x_i y_j} = 
\langle p_i \otimes v_{x_i} , q_j \otimes w_{y_j} \rangle$.
However, the Gram representation is not necessary for 
a useful  \DB, and in this paper, we prove optimal \DB s
for a class of covariance matrices for which no 
Gram representation with a good Gram constant is known.
As will be discussed in Section \ref{sec:fermUV}, these matrices
arise naturally in time--ordered perturbation theory and
standard functional integral representations of 
many--fermion systems.
The constructions we give are
motivated by similar ones in \cite{FKT04}, 
and we shall discuss this relation in more detail in Section \ref{sec:fermUV}.

\begin{satz}
\label{haupt} 
Let $K, k \in \bN_0$, $k+K \ge1$, 
and $C_0, \ldots, C_{k+K}$ be $(\bX \times \bX)$--matrices. 
Assume that for all $l \in \{ 0, \ldots, k+K\}$, 
$C_l $ has a Gram representation with \GC\ $\gamma_l$. 
Let $({\cal J }, \succ)$ be a totally ordered set, 
and for all $l \in \{ 1, \ldots, k+K\}$ let
$\varphi_l$ and $\varphi'_l$ be functions from 
$\bX$ to $\cJ$. 
Denote $1_A =1 $ if $A$ is true and $1_A = 0 $ otherwise. 
Then the $(\bX \times \bX)$--matrix $M$ given by
\begin{equation}\label{Mhaupt}
M_{xy} 
= 
(C_0)_{xy}
+
\sum\limits_{l=1}^k (C_l)_{xy}
1_{\varphi'_l(x) \succ  \varphi_l(y)}
+
\sum\limits_{l=k+1}^{k+K} (C_l)_{xy}
1_{\varphi'_l(x) \succeq  \varphi_l(y)}
\end{equation}
has \DB\ $\dB_M = \sum\limits_{l=0}^{k+K} \gamma_l$. 
\end{satz}

\noindent
Theorem \ref{haupt} is proven in Section \ref{sec:det}.

\medskip\noindent
The bound given in Theorem \ref{haupt} is optimal in the following sense.
Let us assume that for each $l$, the Gram representation for the 
$C_l$ is optimal in that the Gram constant $\gamma_{C_l}$ 
satisfies $\gamma_{C_l}^2 = \sup_{x,y \in \bX} \abs{(C_l)_{xy}}$,
and that the decomposition (\ref{Mhaupt}) is nonredundant 
in the sense that for any choice of $x$ and $y$, only one of the summands is nonzero (in particular, $C_0 =0$).
Then the determinant bound given in Theorem \ref{haupt} 
is optimal up to a factor $k+K$
because 
\begin{equation}
{\dB_M} \le (k+K) \left( \sup\limits_{x,y \in \bX} |M_{xy}| \right)^{\frac12}
\end{equation}
and because, by Definition \ref{det.b}, the \DB\ $\dB_M$ of a matrix $M$ satisfies
\begin{equation}\label{opt}
{\dB_{M}}^{2n} \geq \sup\limits_{x_1 , \dots, x_n \in
\bX \atop y_1 , \dots, y_n \in \bX } \abs{\det (M_{x_i y_j}\delta_{ij})}
= (\sup\limits_{x,y \in \bX} |M_{xy}|)^n.
\end{equation}

\section{The Matsubara UV problem for fermion systems}\label{sec:fermUV}

In this section, we specify the covariances for the many--fermion models,
and then briefly review the well--known problem with the standard
Gram representation due to the slow decay at large frequencies
which is caused by the indicator functions from time ordering, 
which are special cases of  the ones appearing in (\ref{Mhaupt})
(the Matsubara UV problem). 
We show that, if a Gram representation of these
covariances exists, it has rather unusual properties. 
Then we state our main results for these models
which follow directly from our new determinant bounds. 
A detailed analysis of these models will be given in \cite{PeSa}.

We consider the standard many--fermion model, as formulated
for instance in \cite{BratRob} or in \cite{msbook}, Chapter 4. 
The Hamiltonian of this model is of the form $H=H_0 + V$. 
The free part $H_0$ is given by a hopping term (if a lattice model is considered)
or a differential operator (if a continuum model is considered). 
In either case, the relevant data for the present discussion are 
a momentum space $\cB$ dual to configuration space $\Xspace$
and an energy function $E: \cB \to \bR$, 
$\Sp{p} \mapsto E(\Sp{p})$, which assigns an energy $E(\Sp{p})$
to a particle with (quasi)momentum $\Sp{p} \in \cB$.
The interaction part $V$ of $H$ describes the interaction of two or more
particles (see below). 

To be specific, we review briefly how $E$ arises in some relevant cases. 
For a continuum system in $d$ spatial dimensions
without a crystal potential,  $\Xspace = \bR^d$, $\cB = \bR^d$, and 
$E(\Sp{p}) = \Sp{p}^2 - \mu$, where  the parameter $\mu > 0$, the chemical potential, 
is a Lagrange parameter used to adjust the particle density.
Particles in a crystal are modelled by a periodic 
Schr\" odinger operator containing a potential that is $\Gamma$--periodic, 
where $\Gamma \subset \bR^d$ is a lattice of maximal rank. 
In this case,  $\cB$ is the torus $\cB = \bR^d / \Gamma^\#$, 
where $\Gamma^\#$ is the dual lattice to $\Gamma$.
The operator has a band spectrum $p \mapsto (e_\nu (\Sp{p}))_{\nu \in \bN}$, 
with the  index $\nu$ labelling the bands. The case of a single $E$ is obtained by 
restricting to a single band $\nu = \nu_0$
and setting  $E(\Sp{p}) = e_{\nu_0} (\Sp{p}) - \mu$. 
For a (one--band) lattice model on a spatial lattice $\Lambda$, 
$\cB = \bR^d / \Lambda^\#$  is again a torus, 
and $E(\Sp{p})$ is the Fourier transform of 
the hopping matrix (see \cite{msbook}, Chapter 4). 
The motivation for restricting to a single band is that
the interesting case is the one where $E(\Sp{p})$ has a 
nontrivial zero set, and that in many interesting cases, the bands do not
overlap, so that for this zero set, only a single band matters. 

In field theoretic constructions, one often considers 
configuration spaces $\Xspace = \Gamma / L \Gamma$ 
that have sidelengths $L \in \bN$, hence finite volume, 
in which case momentum space is discrete: 
$\cB = \cB_L = L^{-1} \Gamma^\# / \Gamma^\#$. 
We shall consider the cases of finite and infinite volume 
in parallel and use the conventions of \cite{msbook}, Appendix A, 
for the Fourier transform.  We denote by $\mumom (\dd \Sp{p}) $ 
the natural invariant Haar measure on the torus $\cB$; 
specifically, for the continuous torus corresponding to infinite volume 
it is given by $(2\pi)^{-d}$ times Lebesgue measure, for the 
discrete torus $\cB_L$ corresponding to a finite volume
it is given by the inverse of the volume times the counting measure. 
We shall drop the subscript $L$ on $\cB_L$ when no confusion can arise.

The interaction part of the Hamiltonian is assumed to be given 
by a two--body potential $v$, where $v(x-y)$ is the interaction 
energy of a configuration with one particle at $x$ and one particle at $y$.
Most of the present paper is concerned with properties of the covariance, 
in which the interaction plays no role. However, the decay properties
of the interaction are important for convergence of expansions, 
see below. 
The correct treatment of the interaction is difficult, but some progress 
has been made by multiscale expansion methods. One of the purposes
of the present paper is to simplify and extend parts of this analysis, namely 
the ultraviolet (UV) integration, which is quite different from the analysis of
the infrared singularity which arises in the limit of zero temperature. 

We briefly discuss the UV problems arising in such models. 
There is a spatial  UV problem associated to continuum interactions that have
a singularity at coinciding points, such as, for instance, 
a Yukawa potential $\E^{-\alpha |x|}/|x|$, but this is not the issue
we address here.

There are also different UV problems associated to the covariances. 
The first one is related to the noncompactness of momentum space
in the first example mentioned above. A similar problem arises
for the periodic Schr\" odinger operator, namely there is an infinite 
number of bands. For the lattice system, the lattice spacing 
provides a natural spatial ultraviolet cutoff.  The UV problem we
are concerned with here is the discontinuity of the covariance
as a function of the time variable, and the corresponding slow
decay of its Fourier transform in the dual variable, the 
{\em Matsubara frequency}. 
In the continuum case $X=\bR^d$, 
we shall therefore impose a cutoff  on the spatial part of momentum. 
We do this by using the measure $\mumoma$ where, for $a > 0$, 
$\mumoma (\dd \Sp{p})  = \chi (a \Sp{p}) \mumom(\dd \Sp{p})$,
with $\chi$ a nonnegative function on $\bR^d$ of compact support
chosen such that $\int \chi(\Sp{p}) \mumom(\dd \Sp{p}) = 1$,
hence $\mumoma ( \cB ) = a^{-d}$. The UV cutoff parameter $a$ scales similarly
to a lattice spacing: if $\Xspace = a\bZ^d$, $\mumom (\cB)  = a^{-d}$. 
For a general lattice $\Gamma$, which may have different spacings
in the different directions, we define $a$ by $\mumom (\cB ) = a^{-d}$, 
so that $a$ is a geometric mean of the lattice spacings,
and set $\mumoma = \mumom$. 

Let $\beta > 0$, $f_\beta (E) = (1 + \E^{\beta E})^{-1}$, and
for $(\tau,E) \in (-\beta,\beta] \times \bR$ let
\begin{equation}\label{eq:CtaE}
\CC(\tau,E)
=
\cases{
- \E^{ -\tau E} \; (1- f_\beta (E)) & for $0 < \tau \le \beta$ \cr
\E^{-\tau E} f_\beta(E) & for $- \beta  <  \tau \le 0$.}
\end{equation}
Extend the function $\CC$ to a function on $\bR \times \bR$
that is $2\beta$ -- periodic in $\tau$. Note that 
\begin{equation}
\CC (\tau + \beta, E ) 
=
- \CC (\tau, E) .
\end{equation}
In the application, the parameter $\beta $ is the inverse temperature,
and the Fermi function $f_\beta $ is the expected occupation number 
for free fermions. 

\begin{definition}\label{Covdef}
The free covariance (free one--particle Green function) for a many--fermion 
system is the inverse Fourier transform of  the map
$\Sp{p} \mapsto C(\tau, E(\Sp{p}))$ :
\begin{equation}\label{eq:Cttpxxp}
C_{(\tau,\Sp{x}),(\tau',\Sp{x}')}
=
\int_\cB
\mumoma (\dd \Sp{p}) \;
\E^{\I \Sp{p} \cdot (\Sp{x} - \Sp{x}')}
\; 
\CC(\tau-\tau', E(\Sp{p})) .
\end{equation}
More generally, let $\scalefunc \in L^1(\cB, \mumoma)$ and define
\begin{equation}\label{eq:Cscttpxxp}
C_{(\tau,\Sp{x}),(\tau',\Sp{x}')}^{(\scalefunc)}
=
\int_\cB
\mumoma (\dd \Sp{p}) \, \scalefunc (p) \;
\E^{\I \Sp{p} \cdot (\Sp{x} - \Sp{x}')}
\; 
\CC(\tau-\tau', E(\Sp{p})) .
\end{equation}
\end{definition}

\noindent
The function (\ref{eq:Cttpxxp}) arises in time--ordered expansions 
relative to a quasifree state corresponding to a 
quadratic Hamiltonian $H_0$ with dispersion relation $E$, as discussed
above. If  we denote the fermionic field 
operators in a second--quantized formulation by 
$a^{\hphantom{*}}_{\Sp{x}}$  
and set $a^{(+)}_{\tau,\Sp{x}} = \E^{\tau H_0} a^*_{\Sp{x}} \E^{-\tau H_0}$
and 
$a^{(-)}_{\tau,\Sp{x}} = \E^{\tau H_0} a^{\hphantom{*}}_{\Sp{x}} \E^{-\tau H_0}$,
\begin{equation}\label{15}
C_{(\tau,\Sp{x}),(\tau',\Sp{x}')}
=
-    
\omega_0
\left(
\bT
[ 
a^{(-)}_{\tau,\Sp{x}} 
a^{(+)}_{\tau',\Sp{x}'}
]
\right)
\end{equation} 
where $\omega_0$ denotes the quasifree state corresponding to $H_0$
around which we expand, 
and $\bT$ denotes time ordering  \cite{AGD}.
As $\omega_0$ is a KMS state, (\ref{15}) makes sense for all $\tau, \tau' \in \bR$ with $0 \le \abs{\tau-\tau'} \le \beta$. 
Because the field operators obey the canonical anticommutation relations, 
the time ordering, which avoids commutator terms (keeping only the fermionic
antisymmetry), leads to discontinuities in the function, which are explicit in 
(\ref{eq:CtaE}). Thus the discontinuity of $C$ reflects the microscopic structure of the 
physical system, as encoded in the anticommutation relations of the field operators
that generate the observable algebra.  

In the above definitions, we have assumed for simplicity  that
$C_{xy}$ and $V(x,y)$ depend only on space coordinates $x,y \in
\Xspace$, with $\Xspace$ as above. 
It is straightforward to generalize our arguments to the case
with spin or additional indices on which the fields depend
(e.g.\ for the usual models with $SU(N)$ symmetry, this just amounts to
replacing $C$ by $C \otimes 1_N$, where $1_N$ denotes the
$N$--dimensional unit matrix, and the representations by inner
products used below can be adapted in the obvious way by tensoring 
with a factor $\bC^N$ and using that $\delta_{i,j} = \langle e_i, e_j \rangle$
for any orthonormal basis of $\bC^N$).

Obviously, (\ref{eq:Cttpxxp}) can be regarded as defining an $(\bX_d \times \bX_d)$--matrix,
where
\begin{equation}
\bX_d = [0,\beta) \times \Xspace
\end{equation}
Let 
\begin{equation}
\hat{\bX}_d = \Matf \times \cB 
\end{equation}
where $\Matf =  \frac{\pi}{\beta} ( 2 \bZ +1)$.
The Fourier transform of $C$ is 
\begin{equation}
\label{prop.four}
 \hC(\omega, \Sp{p}) = \frac{1}{\I\, \omega -
E(\Sp{p})}, \quad (\omega, \Sp{p}) \in \hat{\bX}_d
\end{equation}
The standard way to obtain a Gram representation for (regularized)
covariances in quantum field theory is via their Fourier representation. 
In our present setting, if $\hD \in L^1(\hat{\bX}_d)$, 
then a Gram representation for $D$ is obtained simply by setting
$\cH = L^2 (\hat{\bX}_d)$, and 
for $(\tau,\Sp{x}) \in \bX_d$, 
\begin{eqnarray}\label{eq:stanGram}
v_{\tau,\Sp{x}} (\omega, \Sp{p}) 
&=& 
\E^{-\I \tau \omega + \I \Sp{p} \cdot \Sp{x}}
\abs{\hD (\omega, \Sp{p})}^{1/2}
\nonumber\\
w_{\tau,\Sp{x}} (\omega, \Sp{p}) 
&=& 
\E^{-\I \tau \omega + \I \Sp{p} \cdot \Sp{x}}
\abs{\hD (\omega, \Sp{p})}^{-1/2}
\;
\hD (\omega, \Sp{p}) .
\end{eqnarray} 
The Gram constant is $\gamma_D = \Vert \hD \Vert_1$, 
and the dominated convergence theorem implies 
continuity of the maps $(\tau,\Sp{x}) \mapsto v_{\tau,\Sp{x}}$ and 
$(\tau,\Sp{x}) \mapsto w_{\tau,\Sp{x}}$.

However, the $\hC$ in (\ref{prop.four}) decays so slowly as a function 
of the Matsubara frequency $\omega$
that $\hC \not\in L^1(\hat{\bX}_d)$
(this must be so because $\CC$ itself has a discontinuity, 
so its Fourier transform cannot be integrable). 
Thus the standard procedure to obtain a Gram representation fails. 

\begin{lemma}
\label{non.sep.1}
Let $U$ be the $(\bR \times \bR)$-matrix given by
\begin{equation}
\label{eq13}
U_{st} = \left\{
\begin{array}{rcc}
1 &,& s \geq t \\
0 &,& s<t
\end{array}
\right. .
\end{equation}
If $(\cH,v,w)$ is a Gram representation of $U$,
then $\cH$ is non-separable and the maps $t \mapsto v_t $ and $t \mapsto w_t$
are discontinuous at all $t \in \bR$. 
\end{lemma}

\begin{proof}
For all $s,t \in \bR$, $U_{st} = \langle v_s, w_t \rangle$, 
so for $t' > t$, $\langle v_t, \; w_t - w_{t'} \rangle  =1$
and for $t' < t$,   $\langle v_{t'}, \; w_t - w_{t'} \rangle  = - 1$.
Thus, by the Schwarz inequality and the  bound $\sup_t \norm{v_t} \le \gamma_U$,
\begin{equation}\label{wtwtp}
\forall t,t': \; t  \ne t'  \Longrightarrow \norm{w_t - w_{t'}} \ge \frac{1}{\gamma_U} .
\end{equation}
Thus the map $t \to w_t $ is discontinuous everywhere. Reversing the roles of 
$v_t$ and $w_t$ in the above argument implies the same for the map $t \to v_t$. 
An obvious variant of this argument implies discontinuity in the weak topology as well. 
Set $ W = \{ w_t : t \in \bR \}$. 
Let $A\subset \cH $ be countable. For all $x \in A$, eq.\ (\ref{wtwtp}) 
and the triangle inequality imply that
$\{ y \in \cH : \norm{y-x} < \frac{1}{4 \gamma_U} \}$ contains at most one element of  $W$.
Thus the $\frac{1}{4 \gamma_U}$ -- neighbourhood of $A$ contains only countably 
many elements of $W$, hence $A$ is not dense in $\cH$.  
\end{proof}

\begin{koro}
\label{non.sep.2}
The covariance matrix of the many--fermion system given by 
(\ref{eq:Cttpxxp}) has no Gram representation on a separable 
Hilbert space. 
\end{koro}

\begin{proof}
The function $\tau \mapsto \DD (\tau,E) = \CC(\tau,E) - \CC (\tau,0)$ is continuous in $\tau$. 
Its Fourier transform, 
\begin{equation}
\omega 
\mapsto
- \; \frac{E}{\I \omega (\I \omega - E)} \; ,
\end{equation}
is in $\ell^1$. Thus 
\begin{equation}
D_{(\tau,\Sp{x}),(\tau',\Sp{x}')}
=
\int
\mumoma (\dd p) \; 
\E^{\I \Sp{p} \cdot (\Sp{x} - \Sp{x}')}
\; 
\DD(\tau-\tau', E(\Sp{p}))
\end{equation}
has the Gram representation given in (\ref{eq:stanGram}).
An elementary argument involving direct sums of Hilbert spaces shows that
$C=D+D'$ has a Gram representation if and only if 
$D'$ has a Gram representations. Assume that $C$, given by 
(\ref{eq:Cttpxxp}), has a Gram representation on a separable Hilbert space $\cH$. 
Then $C - D$ has a Gram representation on a direct sum of separable Hilbert spaces, 
which is itself separable. But $C-D$ is
\begin{equation}
\delta^a_{\Sp{x},\Sp{x}'} 
(U_{\tau,\tau'} - \frac12)
\end{equation}
with $\delta^a_{\Sp{x},\Sp{x}'} 
= 
\int \mumoma (\dd p) \; 
\E^{\I \Sp{p} \cdot (\Sp{x} - \Sp{x}')}$
and $U$ as in  Lemma \ref{non.sep.1}, which has no Gram representation
on any separable Hilbert space. 
\end{proof}

\noindent
Our main use of Gram representations is, of course, to bound determinants 
of the type occurring in (\ref{eq:detbou}).
Lemma \ref{non.sep.1} does not exclude that a useful Gram representation, 
i.e.\ one with a good \GC, can be found, but it shows that the representation
will be very different from the ones used so far in fermion models, which all
involve separable Hilbert spaces and where continuity of the maps $v$ and $w$ holds.  

One can attempt to circumvent the above problem by introducing a UV cutoff $\Omega > 0$, 
which restricts the sum over frequencies $\omega$ to a finite set
(for instance by regularizing to $\hat \CC_\chi (\omega, \Sp{p}) = \hat \CC (\omega, \Sp{p}) \;
\chi(\omega/\Omega)$, where $\chi$ is a smooth function of compact support). 
This obviously makes the standard Gram constant finite, 
Of course, a UV cutoff cannot
simply be imposed, because it implies that the time-ordered
imaginary-time correlation functions are continuous and therefore
not physical. 
The \GC\ $\gamma_{\CC_\chi} \sim \log\Omega$
diverges for $\Omega \to \infty$. 
One can attempt to perform the limit $\Omega \to \infty$ by multiscale and renormalization
techniques. 
The approach via \DB s developed in the next sections is, however, much simpler 
and more natural that such a multiscale approach, and it makes the latter unnecessary.

Recall that momentum space is $\cB = \bR^d$ for an continuous system and
$\cB = \bR^d / \Gamma^\#$ for a system on a lattice $\Gamma$, 
that in the continuum case, $\mumoma$ contains an ultraviolet cutoff, 
and that $\cB_L = L^{-1} \Gamma^\# / \Gamma^\#$ 
is the corresponding momentum space for the finite--volume system. 
The main result about the \DB\ of many--fermion covariances is as follows. 

\begin{satz}\label{nice}
Let $E: \cB \to \bR $ be bounded and measurable.
Then the fermionic covariance matrix 
$C^{(\scalefunc)}$ given in (\ref{eq:Cscttpxxp}) has determinant bound 
\begin{equation}\label{eq:nicedB}
\dB_{C^{(\scalefunc)}} = 2 
\left(\int \mumoma (\dd \Sp{p}) \;  \abs{\scalefunc (\Sp{p})}
\right)^{1/2} .
\end{equation}
In particular, the covariance  $C$ defined  in (\ref{eq:Cttpxxp}) 
has $\dB_C = 2 \mumom^a (\cB)^{1/2}$.
\end{satz}

\noindent 
Theorem \ref{nice} is proven after Corollary \ref{opt2}. 
As mentioned after Theorem \ref{haupt}, this bound is optimal up to the 
prefactor $2$. 

In Section \ref{sec:noUVcutoff}, we discuss the decay constant of these 
covariances and prove a convergence theorem for the expansion for the 
fermionic effective action. 

In Section \ref{sec:largeomega}, we discuss the properties of  covariances
obtained by a splitting into small and large frequencies and prove that the 
integration over fields with large frequencies, which usually is the first 
step in a multiscale treatment, is given by convergent expansions, 
for arbitrarily large initial interaction strength. 

When rewriting traces using Trotter--type formulas, to obtain functional 
integral representations, one typically 
obtains time--discretized covariances. The bounds given here
apply to them as well, uniformly in the parameter $n$ that 
defines the discretization \cite{PeSa}.

\section{Determinants and chronological products}\label{sec:det}

In this section we show that determinants corresponding to a general
chronological  ordering have good determinant bounds
and prove Theorem \ref{haupt}. 
We first recall some standard facts and fix notation. 

\begin{definition}\label{finidi}
Let $V$ be a finite--dimensional vector space over $\bC$.
\begin{enumerate}

\item
Let $k \in \bN$. A totally antisymmetric $k$--linear map $\alpha:
V^k \to \bC$
is called \emph{\EKLF}. 
The vector space of all \EKLF s is
identified with $\BW^k V^*$. We also set $\BW^0 V^* = \bC$.

\item
Let $k,l \in \bN$. The \emph{exterior product} of
$\alpha \in \BW^k V^*$  and $\beta \in \BW^l V^*$,
$\alpha \wedge \beta \in \BW^{k+l} V^*$, 
acts on $v_1, \dots, v_{k+l} \in V$ as
\begin{eqnarray}
&&\hskip-1cm
(\alpha \wedge \beta ) \; (v_1, \dots , v_{k+l})
\\
&=&
\frac{1}{k!l!}
\sum\limits_{\sigma\in S_{k+l}} \mbox{ sgn}\,(\sigma) 
\;
\alpha(v_{\sigma(1)},\dots,  v_{\sigma(k)})\,
\beta(v_{\sigma(k+1)}, \dots, v_{\sigma(k+l)}).
 \nonumber
\end{eqnarray}
Here $S_n$ denotes the set of permutations of $\{ 1, \ldots, n\}$.

The \emph{exterior algebra} $\BW V^*$ over the vector space $V$ is
\begin{equation}
\BW V^*= \bigoplus\limits_{k=0}^\infty \BW^k V^*
\end{equation}
We identify $\BW V $ with $\BW V^{**}$, the exterior algebra over $V^*$.

\end{enumerate}
\end{definition}

\noindent
The following condition defines a duality between the spaces
$\BW^k V^*$ and $\BW^k V$: 
for $\alpha = \alpha_1 \wedge \dots \wedge \alpha_k \in
\BW^k V^*$ and $v = v_1 \wedge \dots \wedge v_k \in
\BW^k V$, 
\begin{equation}
\label{eq23}
\langle \alpha, v \rangle = \det (\alpha_i(v_j))_{i,j} \; .
\end{equation}
This duality defines a vector space isomorphism  
$\BW^k V^* \to (\BW^k V)^*$:
\begin{equation}
\label{eq24}
 \langle \alpha_1 \wedge \dots \wedge \alpha_k , v_1
\wedge \dots \wedge v_k \rangle = \alpha_1 \wedge \dots \wedge
\alpha_k (v_1, \dots, v_k)
\end{equation}
(this isomorphism is unique only up to a multiplicative factor, 
and different conventions are used in the literature). 
Finally, the isomorphisms (\ref{eq24}), $k \in \bN$, canonically 
induce an isomorphism between $\BW V^*$ and $\left(\BW V \right)^*$.

\begin{definition}\label{critters}
Let $\mbox{End} \BW V^*$ denote the set of endomorphisms of $\BW V^*$.

\begin{enumerate}

\item
\label{verjung}
For $w \in \BW V$ define
$w \lrcorner \in \hbox{End} \, \BW V^*$ by the
condition
\begin{equation}
\label{eq25}
\forall v \in \BW V: \quad
\langle w  \lrcorner  \alpha, v \rangle = \langle
\alpha, w\wedge v \rangle.
\end{equation}

\item
\label{anticommut}
For $\alpha \in V^*$ let $(\alpha \wedge)
\in \hbox{End} \, \BW V^*$ be defined by
\begin{equation}
\forall \beta \in \BW V^*: \quad
(\alpha \wedge) : \beta \mapsto \alpha \wedge \beta
.
\end{equation}

\end{enumerate}
\end{definition}

\begin{lemma}
\label{CAR}
These endomorphisms obey canonical anticommutation relations:
\begin{enumerate}
\item{
$(\alpha_1 \wedge) (\alpha_2 \wedge) + (\alpha_2 \wedge) (\alpha_1
\wedge) = 0$, for all $\alpha_1, \alpha_2 \in V^*$. }
\item{
$u_1 \lrcorner u_2 \lrcorner + u_2 \lrcorner u_1 \lrcorner = 0$,
for all $u_1, u_2 \in V$. }
\item{
$(\alpha \wedge) u \lrcorner + u \lrcorner (\alpha \wedge) =
\alpha(u)$, for all $\alpha \in V^*$ and all $u \in V$. }
\end{enumerate}
\end{lemma}

\begin{proof}
Items 1 and 2 are clear. Item 3 holds  because for all $u \in V$, 
$u \lrcorner: \BW^k V^* \to \BW^{k-1} V^* $
is an antiderivation of degree -1: for all 
$\alpha \in \BW^k V^*$ and all $\beta \in \BW V^*$,
$
u \lrcorner (\alpha \wedge \beta) = (u\lrcorner \alpha) \wedge
\beta + (-1)^k \alpha \wedge(u \lrcorner \beta).
$
\end{proof}

\begin{lemma}\label{le33}
Let $n\in \bN$,  $\alpha_1, \dots, \alpha_n \in V^*$ and
$v_1,\dots, v_n \in V$. Then
\begin{equation}
\label{det} \det \Big(\alpha_i(v_j) \Big)_{1 \le i,j \le n} =
(-1)^{\frac{n(n-1)}{2}} v_1 \lrcorner \dots v_n
\lrcorner (\alpha_1 \wedge \dots \wedge \alpha_n).
\end{equation}
\end{lemma}

\begin{proof}
Observe that (\ref{det}) makes  sense because the
right hand side of this equation is an element of $\BW^0 V^*
= \bC$. Eq.\  (\ref{eq25}) implies by induction that
\begin{equation}
v_1\lrcorner \dots v_n \lrcorner(\alpha_1 \wedge \dots \wedge
\alpha_n)
= 
\langle \alpha_1 \wedge \dots
\wedge \alpha_n, v_n \wedge \dots \wedge v_1 \rangle 
\end{equation}
Inverting the order of the $v_i$ and using (\ref{eq23}) gives the claim.
\end{proof}

\begin{definition}\label{ordnung}
Let $(\cJ, \succ)$ be a totally ordered set. 
For $j,j' \in \cJ$, $j\not= j'$ denote
\begin{equation}
\label{eq28}
 1_{j \succ j' } = \left\{
\begin{array}{ccc}
1 &\mbox{ if }& j \succ j' \\
 0 &\mbox{ if }& j' \succ j .
\end{array}
\right.
\end{equation}

\begin{enumerate}

\item
For $J,J' \subset \cJ$ define $\rho(J,J') = (-1)^{N_{J,J'}}$,
where $N_{J,J'}$ is the number of pairs $(j,j') \in J \times J' $
with $j \succ j'$. 

\item\label{piperm}
Let $K \in \bN$ and $\Seq=(j_1, \ldots, j_{K})$ be a finite sequence in $\cJ$, 
such that $k \ne l \Rightarrow j_k \ne j_l$. 
Let $\pi \in S_{K}$ denote the unique permutation chosen such that
for all $k \in \{1, \ldots, K-1\}$, $j_{\pi(k)} \prec j_{\pi(k+1)}$. 
Let $\veps_1, \ldots , \veps_{K} \in \mbox{End }\BW V^*$. 
The {\em $\Seq$--chronological product} of $\veps_1, \ldots , \veps_{K}$ is 
\begin{equation}\label{kronos}
\bT_{\Seq}
\lbrack
\veps_1, \ldots, \veps_{2n}
\rbrack
=
 \mbox{ sgn}\, (\pi) \; 
\prod\limits_{\nu =1}^{2n} \veps_{\pi(\nu)} .
\end{equation}

\item\label{JJp}
Let $J = \{j_1, \ldots , j_n\}$ , $J' = \{ j'_1, \ldots, j'_n\}$ with
$j_1 \prec \ldots \prec j_n$, $j'_1 \prec \ldots \prec j'_n$
and $J \cap J' = \emptyset$. 
Let $\veps_1, \ldots , \veps_{2n} \in \mbox{End }\BW V^*$
and $\Seq = ( j_1, \ldots, j_n, j'_1, \ldots, j'_n )$.
For this special choice we denote 
\begin{equation}\label{quartet}
\bT_{J,J'} 
\lbrack
\veps_1, \ldots, \veps_{2n}
\rbrack
=
\bT_{\Seq} 
\lbrack
\veps_1, \ldots, \veps_{2n}
\rbrack
\end{equation}
and call it the {\em $(J,J')$--chronological product} of $\veps_1, \ldots , \veps_{2n}$.
\end{enumerate}

\end{definition}

\noindent
An obvious consequence is 

\begin{lemma}\label{mischle}
Let$J$ and $J'$ be chosen as in item \ref{JJp}
and $\pi$ as in item \ref{piperm}  of Definition \ref{ordnung}.
Then 
\begin{equation}
\mbox{ sgn}\, (\pi) = \rho (J,J') .
\end{equation}
\end{lemma}

\noindent
This sign is chosen in the definition (\ref{kronos}) 
of the chronological product 
because in our application the $\veps_i$ will be odd elements of the 
graded algebra End $\BW V^*$. In general, the sign involved in the chronological 
product is well--defined only if each $\veps_i$ is either even or odd, and the sign
includes only the permutations of odd elements. 

The main result of this section is the following generalization of
Lemma \ref{le33}.

\begin{satz}
\label{det.rep}
Let $(\cJ, \succ)$ be a totally ordered set 
and $J$ and $J'$ be chosen as in Definition \ref{ordnung}.
For $\alpha_1, \dots, \alpha_n \in V^*$ and $v_1,\dots, v_n \in V$
define the $(n\times n)$--matrix $M$ by 
\begin{equation}\label{Mdef}
M_{kl}
=
\alpha_k (v_l) \; 1_{j'_k \succ j_l} .
\end{equation}
Then 
\begin{equation}\label{det.gen}
\det M 
=
(-1)^{n(n-1)/2} 
\bT_{J,J'} 
\lbrack
v_1\lrcorner , \ldots , v_n \lrcorner , 
(\alpha_1 \wedge), \ldots , (\alpha_n \wedge)
\rbrack
1 .
\end{equation}
\end{satz}

\begin{proof}
Induction on $n$. The case $n=1$ is obvious. 
Let $n \ge 2$ and assume (\ref{det.gen}) to hold for matrices of 
size $n-1$.
By definition and by Lemma \ref{mischle}, the chronological product $\bT_{J,J'} [\ldots]$
on the right hand side of (\ref{Mdef}) is  $\rho (J,J') A_1 \ldots A_{2n}$, 
with $A_i \in \{v_1 \lrcorner, \ldots , (\alpha_n \wedge)\}$.
Suppose that  $A_1 = (\alpha_m \wedge)$ for some $m$.
Then $A_2 \dots A_{2n} 1 = 0$, so 
the right hand side of (\ref{det.gen}) vanishes. 
The indicator function in the definition of $M$ implies that
the $m^{\rm th}$ row of $M$ is zero, so that the left hand side
of (\ref{det.gen}) vanishes,  too. Thus we may assume that $A_1
\in \{ v_1 \lrcorner, \dots, v_n\lrcorner\}$. Because $J$ is ordered,
$A_1= v_1\lrcorner $.
Use
\begin{eqnarray}
A_1 A_2 \ldots A_{2n}
&=&
\sum_{k=2}^{2n} (-1)^k A_2 \ldots A_{k-1} (A_1 A_k + A_k A_1) A_{k+1} \ldots A_{2n}
\nonumber\\
&-&
A_2\ldots A_{2n} A_1 .
\end{eqnarray}
When applied to $1 \in \BW^0 V^*$, the last term vanishes because $A_1 1 = 0$. 
By Lemma \ref{CAR}, $A_1 A_k + A_k A_1= \alpha_m (v_1)$ 
if $A_k= \alpha_m \wedge$ for some $m \in \{ 1, \ldots, n\}$, and 
zero otherwise. 
The position $k$ where $\alpha_m \wedge$ appears in the product is
\begin{equation}
k=1+ |\{j \in J \cup J' : j \prec j'_m\}|
= 
1 + m-1 + |\{j \in J: j \prec j'_m\}|.
\end{equation}
Thus $(-1)^k = (-1)^m \rho (\{ j'_m\},J)$. 
Let $I = J \setminus\{1\}$ and $I'_m=J' \setminus \{j'_m\}$. 
The remaining product $A_2 \ldots A_{k-1} A_{k+1} \ldots A_{2n}$ 
times the sign factor $\rho(I,I'_m)$ 
equals the $(I,I'_m)$--chronological product, so 
\begin{eqnarray}\label{zwi}
&&
\bT_{J,J'} 
\lbrack
v_1\lrcorner \dots v_n \lrcorner 
(\alpha_1 \wedge)\dots (\alpha_n \wedge)
\rbrack 1
=
\nonumber\\
&&
\sum_{m=1}^n 
\sigma_{m}(J,J') \;
\bT_{I,I'_m} 
\lbrack
v_2 \lrcorner \ldots v_n \lrcorner 
(\alpha_1 \wedge)\ldots (\alpha_{m-1}\wedge)\;
 (\alpha_{m+1}\wedge)\ldots(\alpha_n \wedge)
\rbrack 1
\nonumber
\end{eqnarray}
with 
\begin{equation}
\sigma_{m}(J,J') 
=
\rho(J,J')\;
(-1)^m \;
\rho (\{ j'_m\},J)\;
\rho(I,I'_m) 
\end{equation}
By definition,
\begin{equation}
\rho(J,J')
=
\rho(I,I'_m) \;
\rho(J,\{j'_m\})\;
\rho (\{1\}, I'_m) \; ,
\end{equation}
$\rho (\{1\}, I'_m)=1$, and 
\begin{equation}
\rho (\{ j'_m\},J)\; \rho(J,\{j'_m\}) = (-1)^{|J|} = (-1)^n .
\end{equation}
Thus
\begin{equation}
\sigma_{m}(J,J')  = (-1)^{m+ n} .
\end{equation}
The inductive hypothesis applies to the chronological product on the right hand side
of (\ref{zwi}). Combine $(-1)^{n(n-1)/2+m+n}  = (-1)^{(n-1)(n-2)/2} (-1)^{m-1}$. 
The statement of the theorem follows by identifying the right hand side of 
(\ref{zwi}) as the Laplace expansion for the determinant.
\end{proof}

\medskip\noindent
In the remainder of this section, we prepare and give the proof of Theorem \ref{haupt}.

\begin{lemma}\label{le41}
Assume that the space $V$ is a Hilbert space with scalar product 
$\langle \cdot, \cdot \rangle_V$. In
this case we identify $V$ with its dual $V^*$ ($v \in V \mapsto
\langle v, \cdot \rangle_V \in V^*$) and consequently $\BW^k V$
with $\BW^k V^* \cong (\BW^k V)^*$ (see (\ref{eq23}) and
(\ref{eq24})).
\begin{enumerate}
\item{The scalar product $\langle \cdot, \cdot \rangle_V$ of $V$ induces, for each $k\in \bN$, through the
identification of elements of $\BW^k V$ with elements of its
dual $(\BW^k V)^*$ a norm $\norm{\cdot}$ on $\BW^k V$: $\norm{u}^2
= \langle u, u \rangle $. This norm fulfills the parallelogram
identity
\begin{equation}
\norm{u+v}^2 + \norm{u-v}^2 = 2\norm{u}^2 + 2\norm{v}^2, \quad \forall u,v \in
\BW^k V \;,
\end{equation}
hence it defines a compatible scalar product on $\BW^k V$.
Thus $\BW^k V$ and hence $\BW V$ are Hilbert spaces.}
\item{
$(u\lrcorner)^{\dag} = (u \wedge)$ and $(u \wedge)^{\dag}=u\lrcorner$, for all
$u \in V$.}
\item{
$\max\{\|u \lrcorner\|, \|(u\wedge)\| \} \le \norm{u}$, for all $u \in
V$. }
\end{enumerate}
\end{lemma}

\begin{proof}
1.
To see that $\norm{\cdot}$ is nondegenerate, use the defining identity
(\ref{eq23}). The other properties are clear.
Item 2 follows directly from Definition \ref{critters}.\ref{verjung}. 
To see 3, let $u \in V$ and $w \in \BW V$. Then
by Lemma \ref{CAR}
\begin{equation}
\langle w, ( u\lrcorner(u \wedge) + (u \wedge) u\lrcorner ) w
\rangle =  \norm{w}^2\norm{u}^2 .
\end{equation}
Thus 
$
\norm{u}^2 = \sup\limits_{w \in \BW V \atop \norm{w}=1} \langle w, (
u\lrcorner(u \wedge) + (u \wedge) u\lrcorner ) w \rangle \ge 
\max \{\|u\lrcorner\|^2,\| (u \wedge)\|^2\}
$.
\end{proof}

\noindent
In Definition \ref{finidi}, we required the space $V$ to be finite--dimensional, 
to avoid a discussion of subtleties in the relation between $\bigwedge V$ and its dual. 
In our applications, we can always achieve that $V$ is a finite--dimensional subspace 
of a Hilbert space or a reflexive Banach space, by taking $V$ as a space spanned by 
finitely many vectors. For Hilbert spaces, we could alternatively also 
have dropped the condition of finite dimensionality in the above. 

\begin{lemma}
\label{det.bound} Let  $\varphi, \varphi': \bN \to \cJ$ be
functions into a totally ordered set $(\cJ,
\succ)$. Let $\cal H$ be a Hilbert space. For all $n\in \bN$ and
all $v_1, \dots, v_n$, $w_1,\dots, w_n \in {\cal H}$
\begin{equation}
\label{eq44} 
\abs{
\det \Big(\langle v_k,w_l \rangle_{\cal H} \,
1_{\varphi'(k) \succ \varphi(l) } \Big)_{k,l} }
\le
\prod\limits_{k=1}^n \norm{v_k}\;\norm{w_k} .
\end{equation}
The same inequality holds with $1_{\varphi'(k) \succ \varphi(l)}$ replaced by 
$1_{\varphi'(k) \succeq \varphi(l)}$.
\end{lemma}

\begin{proof}
For $n \ge 1$ let $\bN_n = \{1, \ldots , n\}$. Define 
\begin{equation}
\cG_n 
=
\{ j \in \cJ : \exists k,l \in \bN_n: \varphi'(k) = \varphi (l) = j \} .
\end{equation}
Obviously, $|\cG_n| \le n$. Let 
\begin{equation}
m 
=
\max\limits_{j \in \cJ} \{ \abs{(\varphi')^{-1} (\{j\}) \cap \bN_n}, 
\abs{\varphi^{-1}(\{ j\}) \cap \bN_n} \} 
\end{equation}
and set $\tilde \cJ_n = \cJ \times \{ 0,1\} \times \{1,\ldots, m\}$. Extend
the ordering lexicographically, i.e.\ $(j,\mu,\nu) \succ (j',\mu',\nu')$
$\Leftrightarrow$ $j \succ j'$ or [$j = j'$ and $\mu \succ \mu'$] or
[$j=j'$ and $\mu=\mu'$ and $\nu > \nu'$]. Then $(\tilde \cJ_n, \succ)$
is totally ordered. For $j \in \cG_n$, there are $r \le m$ and
$k_1, \ldots, k_r \in \bN_n$ such that  for all $\rho \le r$, $\varphi'(k_\rho) = j$,
and there are $s \le m$, 
$l_1, \ldots, l_s \in \bN_n$ such that  for all $\sigma \le s$, $\varphi(l_\sigma) = j$.
We now extend $\varphi$ to $\tilde \varphi$ and $\varphi'$ to $\tilde\varphi'$
as follows. 

\medskip\noindent
{\em Case of the matrix with $1_{\varphi'(k) \succ \varphi (l)}$. }
In this case, $1_{\varphi'(k) \succ \varphi (l)} = 0 $ if $\varphi'(k) = \varphi (l)$. 
To obtain $1_{\tilde\varphi'(k) \succ \tilde\varphi (l)} = 0 $, 
we make $\tilde\varphi'(k) $ smaller by setting
$\tilde\varphi'(k_\rho) = (\varphi'(k_\rho),0,\rho)$ and
$\tilde\varphi (l_\sigma) = (\varphi(l_\sigma), 1, \sigma)$. 

\medskip\noindent
{\em Case of the matrix with $1_{\varphi'(k) \succeq \varphi (l)}$. }
In this case, $1_{\varphi'(k) \succeq \varphi (l)} = 1 $ if $\varphi'(k) = \varphi (l)$. 
To obtain $1_{\tilde\varphi'(k) \succ \tilde\varphi (l)} = 1 $, 
we make $\tilde\varphi'(k) $ bigger by setting
$\tilde\varphi'(k_\rho) = (\varphi'(k_\rho),1,\rho)$ and
$\tilde\varphi (l_\sigma) = (\varphi(l_\sigma), 0, \sigma)$. 

\medskip\noindent
For $j \in \cJ \setminus \cG_n$, $j=\varphi '(k)$, we set 
$\tilde \varphi ' (k) = (\varphi'(k), 0,\rho)$ etc. By definition of the 
lexicographical ordering on $\tilde \cJ$, it does not matter which 
convention one chooses on $\cJ \setminus \cG_n$. 

By construction, $\tilde\varphi'(\bN_n) = J'$ 
and $\tilde\varphi (\bN_n)= J$ are disjoint, and $|J| =|J'|=n$. 
We may permute the rows and columns of the matrix such that 
$\tilde\varphi(m_1) \prec \tilde\varphi(m_2)$ if $m_1 < m_2$ 
and similarly for $\tilde\varphi'$. This does not change the 
absolute value of the determinant. We can now apply Theorem \ref{det.rep}, 
to represent the determinant as a chronological product. 
The norm estimate in Lemma \ref{le41} implies the statement. 
\end{proof}

\begin{definition}
Let $n \in \bN$ and $A$ be a complex $(n \times n)$--matrix. 
We say that  $\Pi(A,\gamma)$ holds iff
for all $p \in \{ 1, \ldots n\}$ and all sequences
$a_1 < \ldots < a_p$ and $b_1 < \ldots < b_p$ in $\{1, \ldots, n\}$, 
\begin{equation}
\sup_{v_1, \ldots, v_p, w_1, \ldots, w_p \in \ball}
\abs{
\det
\left(
\langle
v_q \, , \, w_r
\rangle\; 
A_{a_q, b_r}
\right)_{1 \le q,r\le p}
}
\le
{\gamma}^{2p} .
\end{equation}

\end{definition}

\begin{lemma}\label{subdetsum}  
Let $n$ and $k \in \bN$ and $A^{(1)} , \ldots, A^{(k)}$ be 
complex $(n \times n)$--matrices. Assume that for all 
$l \in \{1, \ldots, k\}$ there are $\gamma_l > 0$ such that
the property $\Pi(A^{(l)},\gamma_l)$ holds. 
Then $\Pi \left(A^{(1)}+\ldots + A^{(k)}, \; \gamma_1 + \ldots +\gamma_k \right)$ holds. 
\end{lemma}

\begin{proof}
Induction on $k$. For $k=1$, the statement is obvious.
In the induction step, let 
$k \ge 2$, and assume $\Pi \left(A^{(2)}+\ldots + A^{(k)},
\gamma_2 + \ldots +\gamma_k \right)$ 
to hold. 
Let $p \in \{ 1, \ldots n\}$,
$a_1 < \ldots < a_p$, 
and $b_1 < \ldots < b_p$ in $\{1, \ldots, n\}$,
and $v_1, \ldots, v_p, w_1, \ldots, w_p \in \ball$.
Let $B$ and $C$ be the matrices with elements
$
B_{q,r} = 
\langle
v_q \, , \, w_r
\rangle\; 
A_{a_q, b_r}^{(1)} 
$
and 
$
C_{q,r}
=
\langle
v_q \, , \, w_r
\rangle\; 
\sum_{i=2}^k A_{a_q, b_r}^{(i)} 
$
Also, set  
$\gamma'_1 = \sum_{l=2}^k \gamma_l$.
Then by the generalized Laplace expansion for determinants
\begin{eqnarray}
\det(B+C)
=
\sum_{S,T \subset \{1, \ldots , p\} \atop |S|=|T|}
\veps_p(S,T) \; 
\det B_{S,T} \; 
\det C_{S^c,T^c} 
\end{eqnarray}
where $S^c = \{1, \ldots, p \} \setminus S$ and 
$\veps_p(S,T) \in \{ -1,1\}$, and the subscripts 
denote the submatrices of $B$ and $C$ defined by the 
sets.  
Let $s=|S|=|T|$. 
By hypothesis of the Lemma, for all $S$, $T$
\begin{equation}
\abs{\det B_{S,T}} \le \gamma_1^{2s}
\end{equation}
and by the inductive hypothesis, 
\begin{equation}
\abs{\det C_{S^c,T^c}} \le 
{\gamma'_1}^{2(p-s)} .
\end{equation}
Thus, using ${p \choose s}^2 \le {2p \choose 2s}$, 
\begin{equation}
\abs{\det (B+C)} 
\le
\sum_{s=0}^p
{p \choose s}^2
{\gamma_1}^{2s} \;
{\gamma'_1}^{2(p-s)}
\le
\left(\sum_{l=1}^k \gamma_l\right)^{2p}
\end{equation}
\end{proof}

\medskip\noindent
{\it Proof of Theorem \ref{haupt}. }
Call the $n \times n$ submatrices of the summands in (\ref{Mhaupt}) $M_l$.
By Lemma \ref{subdetsum}, it suffices to show that for all 
$l \in \{0, \ldots, k+K\}$,  $\Pi(M_l,\gamma_l)$ holds.
The matrix $C_l$ has a Gram representation $(\cH,g,h)$
with Gram constant $\gamma_l$. Then 
\begin{equation}
\langle v \, , \, w \rangle_{\bC^n} \;
(C_l)_{xy} 
=
\langle v \otimes g_x \, , \, w \otimes h_y \rangle_{\bC^n \otimes \cH}
\end{equation}
and , if $\norm{v} \le 1$, 
$\norm{v \otimes g_x} = \norm{v} \, \norm{g_x} \le \gamma_l$,
similarly for $w \otimes h_y$. $M_l$ is obtained (for $l >0$) 
by multiplying this with an indicator function. 
Every submatrix of $M_l$ is of the same form as $M_l$
and satisfies the hypotheses of Lemma \ref{det.bound}. 
Thus $\Pi(M_l,\gamma_l)$ holds.
\hfill\hbox{\vrule width 7pt depth 0pt height 7pt}
\par\vspace{10pt}

\medskip\noindent
That all submatrices are involved in property $\Pi$,
as necessary for the inductive argument in the proof of Lemma \ref{subdetsum},
is the reason for taking the supremum over the larger set
in Definition \ref{det.b}, instead of taking a supremum over
$P \in \cP_{n,1}$. Submatrices of a $P\in \cP_{n,1}$ are 
in general not positive.
By contrast, the property of having a Gram representation 
on $\bC^n$ with Gram constant 1 is stable under taking 
submatrices.

\section{Convergent expansions without UV cutoffs}
\label{sec:noUVcutoff}

In this section we apply the results of Section \ref{sec:det}  
to the many--fermion covariances introduced in Section
\ref{sec:fermUV}. 
We give explicit determinant and decay bounds, 
and prove Theorem  \ref{nice}.
Moreover, we show that,  
for a multiscale expansions with the standard Fermi surface cutoff
functions and sectorization, 
our results yield  all standard power counting bounds 
without requiring a cutoff on the Matsubara frequencies, 
so that the analytic structure as a function of the frequencies 
can be preserved in such a multiscale analysis. 

\subsection{Determinant bound}
In the following, we apply Theorem \ref{haupt}
to the covariance (\ref{eq:Cscttpxxp}), of which (\ref{eq:Cttpxxp})
is the special case $\scalefunc=1$.
Before stating the details of the representation 
we briefly motivate it. 
By definition, 
\begin{eqnarray}\label{eq:oioi}
\CC (\tau, E) 
&=&
-
1_{\tau > 0}\;
\E^{-\tau E} \; f_\beta ( - E)
+
1_{\tau \le 0} \;
\E^{- \tau E} \; f_\beta ( E)
\end{eqnarray}
Let $\veps > 0$ and 
\begin{equation}\label{eq:Phi}
\Phi (s,\veps) 
=
\frac{1}{\sqrt{\pi}} \;
\frac{\sqrt{\veps \, f_\beta(-\veps)}}{\I s - \veps} .
\end{equation}
Then, since $\veps > 0$, $s \mapsto \Phi(s,\veps) \in L^2 (\bR)$, 
$\norm{\Phi (\cdot, \veps)}_2 \le 1$, and 
\begin{equation}\label{eq:meiomei}
\forall \tau \ge 0, \veps > 0: \quad
\E^{-\veps \tau} \; f_\beta (-\veps) 
=
\int_\bR \dd s\; \E^{\I s \tau} \; \abs{\Phi(s,\veps)}^2 \; .
\end{equation}
Thus, if $\tau = t-t' > 0$, $\E^{-\veps \tau} f_\beta (-\veps) = \langle v_t, v_{t'} \rangle$
with $v_t (s) = \E^{-\I s t} \Phi (s,\veps)$. 
To use this for $\CC$ we need to respect the signs in (\ref{eq:oioi}), 
hence rewrite, for $\tau \in [-\beta,\beta]$
\begin{eqnarray}\label{plussi}
\CC (\tau, E)
&=&
\cases{
-
\E^{-\tau E} \; f_\beta (-E) & if $\tau > 0$ and $E > 0$ \cr
-
\E^{(\beta-\tau) E} \; f_\beta (E) & if $\tau > 0$ and $E < 0$ \cr
\E^{-(\beta+\tau) E} \; f_\beta (-E) & if $\tau \le 0$ and $E > 0$ \cr
\E^{-\tau E} \; f_\beta (E) & if $\tau \le 0$ and $E < 0$ 
}
\end{eqnarray}
using $f_\beta (-E) = \E^{\beta E} f_\beta (E)$. 
By Tonelli's theorem and an obvious decomposition of the remaining factors
in the integrand, we can represent  $C_{(t,x),(t',x')}$
by integration over $\Sp{p}$. Note that the $v_t$ defined above vanishes
at $E=0$, but that $\CC(\tau,0) =  \frac12 - 1_{\tau >0} \ne 0$, 
so it is necessary to restrict to functions $E(\Sp{p})$ whose zero level set 
has measure zero.

\begin{lemma}\label{le:5.4}
Let $E: \cB \to \bR $ be measurable 
and assume that 
\begin{equation}\label{measure0}
\mumoma\left( \{\Sp{p} \in \cB: E(\Sp{p}) = 0\}\right) = 0 \; .
\end{equation}
Let $\scalefunc \in L^1(\cB,\mumoma)$ with $\scalefunc(\Sp{p}) \ge 0$ for all $\Sp{p} \in \cB$.
For $x = (t,\Sp{x}) \in \bX_d$ and $\sigma \in \{ -1,1\}$ define
\begin{eqnarray}
g^{\sigma}_{x} (s,\Sp{p}) 
&=&
\E^{-\I \Sp{p} \cdot \Sp{x} - \I s t} \; 
\Phi \left(s, |E(\Sp{p})|\right) \;  
\sqrt{\scalefunc(\Sp{p})}
1_{\sigma E(\Sp{p}) > 0}
\nonumber\\
h_{x} (s,\Sp{p}) 
&=&
\E^{-\I \Sp{p} \cdot \Sp{x} + \I s t} \; 
\Phi \left(s, |E(\Sp{p})|\right) \;
\sqrt{\scalefunc(\Sp{p})}
1_{E(\Sp{p}) < 0} \; .
\end{eqnarray}
Then for all $x \in \bX_d$,  $g^+_{x}$, $g^-_x$ and $h_x$ are 
in $\cH = L^2( \bR \times \cB, \dd s \otimes \dd \mumoma)$, 
with norms bounded by $\norm{\scalefunc}_1^{1/2}$,
and the covariance (\ref{eq:Cscttpxxp}) has the representation
\begin{eqnarray}\label{eq:goodrep}
C_{(t,x),(t',x')}^{(\scalefunc)}
&=&
1_{t > t'}\;
\langle
- g^+_{t,\Sp{x}} - g^-_{\beta-t,\Sp{x}}\, , \; g^+_{t',\Sp{x}'}+ h_{t',\Sp{x}'}
\rangle
\nonumber\\
&+&
1_{t \le t'} \;
\langle
g^+_{t,\Sp{x}} + h_{t,\Sp{x}} \, , \; g^+_{t'-\beta,\Sp{x}'}+h_{t',\Sp{x}'}
\rangle \; .
\end{eqnarray}
\end{lemma}

\begin{proof}
The integrand in (\ref{eq:Cscttpxxp}) is bounded, so we can remove the 
set of measure zero $\{ \Sp{p} \in \cB : E(\Sp{p}) =0\}$
from the integral. 
On its complement, the Gram representation given in the lemma
converges absolutely as an iterated integral first over $s$, then over $\Sp{p}$, 
hence by Tonelli's theorem in any order of integration,
and the $L^2$--norms are finite by the same argument.
The bound for the norms is obvious from the properties of 
$\Phi$.  By the support properties of the functions,
\begin{equation}
\langle
- g^+_{t,\Sp{x}} - g^-_{\beta-t,\Sp{x}}\, , \; g^+_{t',\Sp{x}'}+ h_{t',\Sp{x}'}
\rangle
=
\langle
- g^+_{t,\Sp{x}} \, , \; g^+_{t',\Sp{x}'}
\rangle
+
\langle
- g^-_{\beta-t,\Sp{x}} \, , \; h_{t',\Sp{x}'}
\rangle
\end{equation}
and
\begin{equation}
\langle
g^+_{t,\Sp{x}} + h_{t,\Sp{x}} \, , \; g^+_{t'-\beta,\Sp{x}'}+h_{t',\Sp{x}'}
\rangle 
=
\langle
g^+_{t,\Sp{x}} \, , \; g^+_{t'-\beta,\Sp{x}'}
\rangle
+
\langle
h_{t,\Sp{x}} \, , \; h_{t',\Sp{x}'}
\rangle
\end{equation}
Decomposing the integration domain into 
$\cB_\pm
=
\{ \Sp{p} \in \cB \; : \;
\pm E(\Sp{p}) > 0 \}$, 
(\ref{eq:goodrep}) follows from (\ref{eq:meiomei}) and (\ref{plussi}).
\end{proof}

\noindent
The condition that $\scalefunc \ge 0$ in Lemma \ref{le:5.4} was just for 
convenience in stating the result in a simple form. With an obvious 
generalization, replacing $\sqrt{\scalefunc(\Sp{p})}$ by 
$\scalefunc(\Sp{p})\; \abs{\scalefunc(\Sp{p})}^{-1/2}$,
and defining a few more functions $\tilde g$ to take care of the necessary complex
conjugations, a representation with the same properties as  
(\ref{eq:goodrep}) can be obtained for general $\scalefunc \in L^1 (\cB, \mumoma)$. 
In the applications below, $\scalefunc$ will be a scaling function, hence nonnegative. 

\begin{koro}\label{opt2}
Under the hypotheses of Lemma \ref{le:5.4}, 
the many--fermion covariance (\ref{eq:Cscttpxxp}) 
has a \DB\ $\dB_{C^{(\scalefunc)}}$ with
\begin{equation}
\frac{1}{\sqrt{2}} \norm{\scalefunc}_1^{1/2} \le \dB_{C^{(\scalefunc)}} \le 2 \norm{\scalefunc}_1^{1/2}
\end{equation}
(for $\scalefunc = 1$, corresponding to the covariance (\ref{eq:Cttpxxp}), 
$\norm{\scalefunc}_1 = \mumoma (\cB) = a^{-d}$).
\end{koro}

\begin{proof}
The indicator functions in the times $t$ and $t'$ correspond to the 
choices $(\cJ,\succ) = ([-\beta,\beta], >)$,  
$\varphi_1 (t,\Sp{x}) = \varphi'_1 (t,\Sp{x}) = t $ 
and $\varphi_2 (t,\Sp{x}) = \varphi'_2 (t,\Sp{x}) = -t$. 
The upper bound follows from the explicit representation given in  
Lemma \ref{le:5.4} by applying Theorem \ref{haupt}. 
Let 
\begin{equation}
\rho_\pm
=
\int_\cB \mumoma(\dd \Sp{p}) \; 
f_\beta (\pm E(\Sp{p})) \;
\scalefunc(\Sp{p}) 
\end{equation}
then $\rho_- = \norm{\scalefunc}_1 - \rho_+$.
Set $\Sp{x} = \Sp{x}'$. Then considering the cases
$t=t'$ and $t' \uparrow t$ gives
\begin{equation}
\sup_{x,x' \in \bX_d} \abs{C^{(\scalefunc)}_{xx'}} 
\ge
\max \{ \rho_+,\rho_- \} 
\ge \frac12 \norm{\scalefunc}_1 .
\end{equation}
The lower bound for $\dB_{C^{(\scalefunc)}}$ now follows from (\ref{opt}).
\end{proof}

\medskip\noindent
{\it Proof of Theorem \ref{nice}. } 
To apply  Lemma \ref{le:5.4}, we need to satisfy the zero measure condition.
For $\veps > 0$, define $E_\veps : \cB \to \bR$ by $E_\veps (\Sp{p}) = \veps/2$
if $\abs{E(\Sp{p})} \le \veps/2$ and  $E_\veps (\Sp{p}) = E(\Sp{p}) $ otherwise. 
Obviously,  $\norm{E - E_\veps}_\infty \le \veps$, 
and $ \{\Sp{p} \in \cB: E_\veps (\Sp{p}) = 0\} =\emptyset$. 
Because $\beta < \infty$, the covariance $C^{(\scalefunc)}$ is a continuous 
function of $E$ in $\norm{\cdot}_\infty$, so $C^{(\scalefunc)}$ is the limit
$\veps \to 0$ of the covariance $C^{(\scalefunc,\veps)}$ given by $E_\veps$.
By construction,  $E_\veps$ satisfies the conditions of Lemma \ref{le:5.4}
so Corollary  \ref{opt2} implies the bound  (\ref{eq:nicedB}) for $C^{(\scalefunc,\veps)}$.
That bound is uniform in $\veps$.
\hfill\hbox{\vrule width 7pt depth 0pt height 7pt}
\par\vspace{10pt}

\medskip\noindent
The representation of $C^{(\scalefunc)}$ given in Lemma \ref{le:5.4}
generalizes one found in \cite{FKT04}, 
where determinants of matrices of the form
$$
M_{kl}= \langle v_k, w_l\rangle \left\{
\begin{array}{ccc}
0 &,& t_k-t_l \le 0 \\
e^{-(t_k -t_l )} &,& t_k -t_l > 0
\end{array}
 \right.,
$$
for vectors $v_k, w_l$ in a Hilbert space $\cal H$ and
real numbers $t_k, t_l$, were considered. 
The result of \cite{FKT04} corresponds to the special case of the function 
\begin{equation}
\tilde \CC (\tau ) 
=
-
\E^{-\tau } 1_{\tau > 0},
\end{equation}
which is the limit $\beta \to \infty$ of (\ref{eq:oioi}) at $E=1$. 
Thus our method applies to that case, with $\tilde \Phi (s) = (\I s -1)^{-1}$.

\subsection{Decay constant}
Under very mild conditions on $E$, the \DB s we have proven
are uniform in $\beta$ (see Corollary \ref{opt2}). 
One must of course not jump to the conclusion that this implies 
convergence of perturbation series uniformly in the temperature 
because a finite \DB\ is only one condition for convergence of the 
perturbation expansion. The second is the finiteness of the 
decay constants
\begin{equation}
\abf_C^{(k_0,k)}
=
\int_{-\beta}^{\beta} \dd \tau \int_{\Xspace}\dd \Sp{x} \;
\abs{C(\tau,\Sp{x})} \;
\abs{\tau}^{k_0} \abs{x}^k
\end{equation} 
for $k_0 \ge 0$ and $k \ge 0$. 
In this paper, we only discuss the case
$k_0 = k=0$, and denote $\abf^{(0,0)}_C = \abf_C$ 
because the simplest convergence theorem requires only this 
data, and because the generalization is straightforward. 
For our many--fermion covariance, the existence
of a nonempty Fermi surface that is not degenerated to a point 
implies that the decay constant grows polynomially in $\beta$
and diverges in the zero--temperature limit.  
Only for special situations, such as a model for an insulator, 
for which $|E(\Sp{p})| \ge \Emin  > 0$, 
the decay constant is uniform in $\beta$. 

For simplicity we assume here the case of a continuous torus $\cB$.
The case of a discrete torus corresponding to a finite volume is similar,
and treated in \cite{PeSa}.

\bigskip\noindent
For $z\in \bC$ and $\veps \ge 0$ set $\tnorm{z}{\veps} = \max \{ |z| , \veps \}$.

\begin{lemma}\label{le:decayonuv}
Let $E \in C^{d+2}(\cB,\bR)$.
Let $0< \epsilon < 1 $ and assume that 
$
\scalefunc (\Sp{p})
=
\scf (\frac{E(\Sp{p})}{\epsilon})
\scg (\Sp{p})
$
where $\scf \in C^\infty (\bR, \bR_0^+)$ and $\scg \in C^\infty (\cB, \bR_0^+)$. 
Let $\MI\in \bN_0^{d}$ be a multiindex and $\ami=|\MI|$.

\begin{enumerate}
\item
There is a constant $\Kd > 0$ such that for $\ami \le d+1$, 
then
\begin{equation}\label{Kditem}
\int_{-\beta}^\beta \dd \tau \;
\abs{
\Sp{x}^{\MI} \;
C^{(\scalefunc)}_{(\tau,\Sp{x}),(0,0)}
}
\le
\Kd
\sum_{m=0}^{\ami}
\epsilon^{m-\ami} \;
\int\limits_{{\rm  supp \;}\scalefunc}
\frac{\mumoma (\dd \Sp{p})}%
{\tnorm{E(\Sp{p})}{\frac{1}{\beta}}^{m+1}}
\end{equation}

\item
If there is $\kappa_0 > 0$ such that 
for all $E$ for which 
$\hat S_{E,\scg} = \{ \Sp{p}\in {\rm supp}\; \scg : E(\Sp{p}) =E\}$ is nonempty, 
$\inf_{\Sp{p} \in \hat S_{E,\scg}} \abs{\nabla E(\Sp{p}) } \ge \eta > 0$, 
and the submanifold $\hat S_{E,\scg}$ of $\cB$
has Gauss curvature bounded below 
pointwise by $\kappa_0$, 
then  there is a constant $\tKd > 0$ such that
for $\ami \le \lceil \frac{d+2}{2} \rceil$
\begin{equation}\label{tKditem}
\int_{-\beta}^\beta \dd \tau \;
\abs{
\Sp{x}^{\MI} \;
C^{(\scalefunc)}_{(\tau,\Sp{x}),(0,0)}
}
\le
\frac{\tKd}{|\Sp{x}|^{\frac{d-1}{2}}}
\;
\sum_{m=0}^{\ami}
\epsilon^{m-\ami} \;
\int \dd E \;
\frac{1_{\frac{E}{\epsilon} \in {\rm supp}\, f}}%
{\tnorm{E}{\frac{1}{\beta}}^{m+1}}
\end{equation}

\end{enumerate}

\end{lemma}

\begin{proof}
We have
\begin{equation}
\Sp{x}^\MI
C^{(\scalefunc)}_{(\tau,\Sp{x}),(0,0)}
=
\int
\mumoma (\dd\Sp{p}) \; 
\CC (\tau , E(\Sp{p})) \; \scalefunc (\Sp{p})\;
\left(- \I \sfrac{\partial}{\partial \Sp{p}}\right) ^{\MI} \;
\E^{\I \Sp{p} \cdot \Sp{x}} .
\end{equation}
Upon integration by parts, the derivative can act in four places --- 
on $\CC$, on either of  the factors $\scf$ and $\scg$ in $\scalefunc$, 
or (for the continuum system)
on the spatial ultraviolet cutoff function $\chi$ in 
$\mumoma (\dd \Sp{p}) = \chi (a\Sp{p}) \dd \Sp{p}$. 
Thus
\begin{eqnarray}
\Sp{x}^\MI
C^{(\scalefunc)}_{(\tau,\Sp{x}),(0,0)}
&=&
\sum_{m=0}^{\ami}
\sum_{n=0}^{\ami-m}
\epsilon^{-n} \;
\int
\mumoma (\dd\Sp{p}) \; 
\Gamma_m (\tau, E(\Sp{p})) \;
\scf^{(n)} \left(\sfrac{E(\Sp{p})}{\epsilon}\right)\;
G^{(\MI)}_{m,n} (\Sp{p})
\; \E^{\I \Sp{p} \cdot \Sp{x}}
\nonumber
\end{eqnarray}
where $G^{(\MI)}_{m,n} \in C^{d+2-\ami} (\cB, \bR)$
is independent of $\epsilon$ and  
satisfies supp $G^{(\MI)}_{m,n} \subset$ supp $\scg$,
and 
\begin{equation}
\Gamma_m (\tau, E)
=
\frac{\dd^m}{\dd E^m}
\CC (\tau, E) \; .
\end{equation} 
Taking the absolute value inside all sums and integrals
and using that 
\begin{equation}\label{Gamint}
\int_{-\beta}^\beta \dd \tau \; 
\abs{\Gamma_m (\tau, E) }
\le
\const {\tnorm{E}{\frac{1}{\beta}}}^{-m-1}, 
\end{equation}
we obtain (\ref{Kditem}).
To prove (\ref{tKditem}), we 
rewrite
\begin{eqnarray}
\Sp{x}^\MI
C^{(\scalefunc)}_{(\tau,\Sp{x}),(0,0)}
&=&
\sum_{m=0}^{\ami}
\sum_{n=0}^{\ami-m}
\epsilon^{-n} \;
\int \dd E \;
\Gamma_m (\tau, E) 
\scf^{(n)} \left(\sfrac{E}{\epsilon}\right)\;
S_{E,G^{(\MI)}_{m,n}} (\Sp{x})
\end{eqnarray}
where
\begin{equation}
S_{E,G^{(\MI)}_{m,n}} (\Sp{x})
=
\int \mumoma (\dd\Sp{p}) \; 
\delta (E - E(\Sp{p}))\;
G^{(\MI)}_{m,n} (\Sp{p})
\; \E^{\I \Sp{p} \cdot \Sp{x}}
\end{equation}
By standard theorems about the Fourier transform of 
surfaces \cite{Stein}, 
\begin{equation}
\abs{
S_{E,G^{(\MI)}_{m,n}} (\Sp{x})
}
\le
\const |\Sp{x}|^{-\frac{d-1}{2}}
\end{equation}
with a constant that depends on $\kappa_0$ 
and $E$, and which is finite under our regularity 
assumption on $E$. 
Finally, we use again (\ref{Gamint}).
\end{proof}

\medskip\noindent
The regularity assumptions on $E$ in Lemma \ref{le:decayonuv} are not optimized. 
For improved bounds using smoothing techniques,  
see \cite{PeSa}.
The scaling function $\scalefunc$ can be chosen $C^\infty$ in our applications, 
so that the assumptions of Lemma \ref{le:decayonuv} on $\scalefunc$ are not 
restrictive. 

This Lemma allows us to bound decay constants as follows.

\begin{koro}\label{koko}
Let $E \in C^{d+2} (\cB,\bR)$. 

\begin{enumerate}

\item
$\abf_{C} \le \const \beta^{d+1}$.

\item
If the system is an insulator, i.e.\ if 
there is $E_0 > 0$ such that for all $\Sp{p} \in \cB$, 
$\abs{E(\Sp{p})} \ge E_0$, then 
\begin{equation}
\abf_{C^{(\scalefunc)}}
\le \const
E_0^{-d-1}
\end{equation}
The constant is proportional to the volume of the support 
of $\scalefunc$. For $\scalefunc = 1$, it is proportional
to $\mumoma(\cB)$.  

\end{enumerate}
\noindent
If there is $E_1$ such that for all energies $E$ 
with $|E| \le E_1$ the level sets satisfy the 
hypotheses of Lemma \ref{le:decayonuv},
item 2, then we also have:

\begin{enumerate}
\item[3.]
\begin{equation}\label{bello}
\abf_{C} 
\le 
\const
\left(
E_1^{-d-1} + \beta^{\frac{d+3}{2} }
\right)
\end{equation}

\item[4.]
If $f(x) = 0$ unless $ 1 \le |x| \le 2$, then 
\begin{equation}\label{4}
\abf_{C^{(\scalefunc)}}
\le
\const
\epsilon^{-\frac{d+1}{2}}
\end{equation}

\item[5.]
For a sector of angular radius $\sqrt{\epsilon}$, 
i.e.\ $g(p) = \gamma (\frac{p}{\sqrt{\epsilon}})$, 
with $\gamma$ supported near $0$, 
$\abf_{C^{(\scalefunc)}}
\le
\const
\epsilon^{-1}
$ .

\end{enumerate}

\end{koro}

\begin{proof}
The first bound follows by the standard summation argument from 
$\tnorm{E(\Sp{p})}{\frac{1}{\beta}} \ge \frac{1}{\beta}$.
The case of an insulator follows immediately from 
$\tnorm{E(\Sp{p})}{\frac{1}{\beta}} \ge E_0$. 
To prove (\ref{bello}), we insert a partition of unity 
$\chil (\frac{E(\Sp{p})}{\epsilon})
+\chig (\frac{E(\Sp{p})}{\epsilon}) =1$, 
where $\chil (x)$ vanishes for $|x| \ge 1$.
The support condition on $f$ in item 4 implies
$\tnorm{E(\Sp{p})}{\frac{1}{\beta}} \ge \epsilon$. 
Again, summation implies the result. The sector estimate
is similar. 
\end{proof}

\subsection{Convergence theorem}
In the following we state a theorem about convergence of expansions 
for the effective action which generalizes the main theorem of 
\cite{Salm99}. As in \cite{Salm99}, we define an interaction $V$ by its
interaction vertices $v_{\barm,m}: \bX^{\barm} \times \bX^m \to \bC$ as 
\begin{equation}\label{Vrep}
V(\Psi ) 
= 
\sum_{m,\barm \ge 0} \int \dd^{\barm} \ul{X} \dd^{m}
\ul{X}' v_{\barm,m} (\ul{X},\ul{X}') \psq^{\;\barm} (\ul{X}) \ps^{m}
(\ul{X}')
\end{equation}
where $\ul{X} = (X_1,\ldots X_m)$ and $\ps^m(\ul{X}) = \ps(X_1)
\ldots \ps(X_m)$. 
For $h>0$, let
\begin{equation}\label{normdef}
\norm{V}_h =  \sum_{m,\barm\ge 0 \atop m+\barm \ge 1} |v_{\barm, m}| h^{\barm+m}
\end{equation}
where 
\begin{equation}\label{stanley}
|v_{\barm, m}| = \max_{i\in \bN_{\barm+m}}\; \sup_{X_i} \int
\prod\limits_{j \ne i} \dd X_j \; |v_{\barm, m}(X_1,\ldots,X_{\barm+m})| .
\end{equation}

\begin{satz}\label{konv}
Let $C$ be an $(\bX \times \bX)$--matrix, considered as a covariance
for a fermionic Gaussian integral, with finite \DB\ $\dB_C$ and  decay bound
$\abf_C$. Denote $\omega_C = 2 \abf_C \dB_C^{-2}$. 
Let $h > 0$, $h'= h + \omega_C$, and let $V$ be an interaction with $\norm{V}_{h'} < \infty$.
Then the effective action  $W(V,C)$, defined as
\begin{equation}
W(V,C) 
=
\log \int \dd\mu_C(\Psi') \; \E^{V(\Psi'+\Psi)} ,
\end{equation}
exists and is analytic in $V$: 
let $W(V,C) = \sum_{p \ge 1} \frac{1}{p!} W_p(V,C) $ be the expansion of $W$ in
powers of $V$. Then for all $P \ge 1$,
\begin{equation}\label{orderP}
\norm{W(V) - \sum_{p=1}^{P} \frac{1}{p!} W_p(V,C)}_h
\le
{\omega_C}^P \frac{{\norm{V}_{h'}}^{P+1}}{1- \omega_C \norm{V}_{h'}}.
\end{equation}

\end{satz}

\begin{proof}
Same as in \cite{Salm99}, except that in the bound for the determinants, 
Lemma 6 of \cite{Salm99}, the Gram constant  is replaced by the \DB\ $\dB_C$. 
\end{proof}

\noindent
The coefficients in the expansion of $W(V,C)(\Psi)$ in the fields $\Psi$ 
are the amputated connected Green functions, so the above theorem 
implies their analyticity in the interaction. In particular, analyticity holds 
for all cases listed in Corollary \ref{koko}, with the appropriate constants. 
In case of an insulator, the convergence radius is uniform in the temperature. 
In case of scaled propagators, one obtains power counting bounds that
are on all scales operationally equivalent to those with a frequency cutoff. 
That no $\omega$-cutoff is needed implies that the analytic structure
as a function of $\omega$ need not be mutilated in a multiscale construction.

\section{Bounds for the integration over large frequencies}
\label{sec:largeomega}

In a multiscale analysis of many--fermion systems, the integration 
over fields with large Matsubara frequency is often the first integration
step in the analysis. In the following we give bounds for the 
effective action obtained by this integration step. 
We first decompose the covariance $C^{(\scalefunc)}$ given in (\ref{eq:Cscttpxxp}) 
in an ultraviolet and an infrared part.

Let $\chil$ and $\chig \in C^\infty (\bR,[0,1])$ with $\chil  + \chig  = 1$, 
$\chil (0) = 1$,   
with constants $\Kchil > 0$ and $\alpha > 0$ such that  
$\chil(x) \le  \Kchil |x|^{-\alpha}$ for all $|x| \ge 1$. 
Abbreviate the covariance $C_{(\tau,\Sp{x}),(\tau',\Sp{x}')}^{(\scalefunc)}
= C^{(\scalefunc)} (\tau'-\tau, \Sp{x}'-\Sp{x})$.
The covariance
\begin{equation}
\Comhl (\tau, \Sp{x}) 
=
\frac{1}{\beta}
\sum_{\omega}
\int \mumoma(\dd \Sp{p}) \; 
\E^{-\I \omega \tau + \I \Sp{p} \cdot \Sp{x}} \; 
\chil \left(\sfrac{\omega}{\Omega}\right) \;
\frac{\scalefunc (\Sp{p})}{\I \omega - E(\Sp{p})}
\end{equation}
is the infrared part of $C^{(\scalefunc)}$,
and 
\begin{equation}\label{IRUVdecomp}
\Comhg (\tau, \Sp{x}) 
=
C^{(\scalefunc)} (\tau, \Sp{x}) - \Comhl(\tau, \Sp{x})  .
\end{equation}
is the ultraviolet part of $C^{(\scalefunc)}$. 
An obvious variant of this decomposition is one where the argument of the function 
$\chil$ is $\Omega^{-2} (\omega^2 + E(\Sp{p})^2)$. Our bounds adapt to 
this choice in an obvious way, so we will not discuss it further here. 

By standard properties of Grassmann Gaussian integration, 
the convolution with the Gaussian measure $C^{(\scalefunc)}= \Comhg+\Comhl$
becomes an iterated convolution, first with $\Comhg$, then 
with $\Comhl$ (see, e.g.\ \cite{msbook}).

\subsection{Determinant bound}

\begin{lemma}
Let $\chil$ be chosen as above, $\Omega \ge 1$, and $\beta > \pi$.
Let $E$ be continuous. Then the \DB\ of $\Comhg$ 
satisfies
\begin{equation}
\dB_{\Comhg}^2
\le
\norm{\scalefunc}_1 \;
(\horror +  2 \ln \Omega) 
+
\int\limits_{|E(\Sp{p})| \le 1} \mumoma (\dd \Sp{p}) \;
\abs{\scalefunc (\Sp{p})} \;
\ln \frac{1}{\tnorm{E(\Sp{p})}{\frac{\pi}{\beta}}} \; 
\end{equation}
where 
$\horror = 10 + 2\Kchil (\alpha^{-1} + (\beta\Omega)^{-1})$.
\end{lemma}

\begin{proof}
By our hypothesis on the decay of $\chil$, the Fourier transform of the 
covariance $\Comhl$ is $\ell^1$ in the Matsubara frequency. Thus $\Comhl$ has 
a Gram representation of type (\ref{eq:stanGram}), with finite Gram constant
$\gamma_<$. By (\ref{IRUVdecomp}) and Theorem \ref{haupt}, a \DB\ for $\Comhg$ is given by 
$\dB_{C^{(\scalefunc)}} + \gamma_<$. $\dB_{C^{(\scalefunc)}}$ 
was bounded in  Corollary \ref{opt2}, 
so it suffices to estimate $\gamma_<$. By definition, 
\begin{equation}
{\gamma_<}^2
=
\frac{1}{\beta}
\sum_{\omega\in\Matf}
\int \mumoma (\dd \Sp{p}) 
\chil (\sfrac{\omega}{\Omega}) \;
\frac{\abs{\scalefunc (\Sp{p})}}{\abs{\I \omega - E(\Sp{p})}} \; .
\end{equation}
The contribution from $|\omega| \ge 1$ is bounded by 
$$
\norm{\scalefunc}_1 \; 
\frac{1}{\beta}\sum_{\omega\in\Matf}
\chil (\sfrac{\omega}{\Omega}) \;
\frac{1}{\abs{\omega }} 
\le
2 \norm{\scalefunc}_1 \;
\left(
\sfrac{1}{\pi} + \ln \Omega 
+
\frac{\kappa}{\alpha} 
+
\frac{\kappa}{\beta\Omega}
\right).
$$
For the contribution from $|\omega| < 1$, we will repeatedly use the elementary bound
$
\frac{1}{\beta}
\sum_{\omega\in\Matf}
1_{|\omega| < u} \le \frac{2 u}{\pi}
$.
For $\abs{E(\Sp{p})} \ge 1$, $|\I \omega - E(\Sp{p}) |^{-1} \le 1$, so 
the contribution from $|\omega| < 1$ and $\abs{E(\Sp{p})} \ge 1$
is bounded by $2 \norm{\scalefunc}_1  /\pi$. 
For $\abs{E(\Sp{p})} \le 1$, we use that 
\begin{equation}
\frac{1}{\beta}
\sum_{\omega\in\Matf}
1_{|\omega| < \abs{E(\Sp{p})}} 
\frac{1}{\abs{\I \omega - E(\Sp{p})}}
\le
\frac{1}{\beta \abs{E(\Sp{p})}}
\sum_{\omega\in\Matf}
1_{|\omega| < \abs{E(\Sp{p})}} 
\le 
\frac{2}{\pi}
\end{equation}
and, bounding the sum by an integral,
\begin{equation}
\frac{1}{\beta}
\sum_{\omega\in\Matf \atop \abs{E(\Sp{p})} \le |\omega| \le 1}
\frac{1}{\abs{\I \omega - E(\Sp{p})}}
\le 
\frac{1}{\beta}
\sum_{\omega\in\Matf  \atop \abs{E(\Sp{p})} \le |\omega| \le 1}
\frac{1}{|\omega|}
\le
\frac{2}{\pi}
+
\ln \frac{1}{\tnorm{E(\Sp{p})}{\frac{\pi}{\beta}}} \; .
\end{equation}
\end{proof}

\subsection{Decay constant}
In this section we show that for a strict cutoff function $\chil$, 
and under natural assumptions on the function $E$, 
the decay constant of $\Comhg$ is bounded by a multiple of $\Omega^{-1}$.
Thus the extra factor $\log \Omega$ from the \DB\ can be avoided  in this bound.

\begin{lemma}
Assume that $\chil$ satisfies $\chil (x) = 1 $ for $|x| \le 1$ and $\chil (x) =0$
for $|x| \ge 2$. Let $\Omega \ge 1$. Assume that the dispersion function $E$ is 
the Fourier transform $E = \hat F$  of some $F \in L^1 (\Gamma, \bC)$, 
and that the inverse Fourier transform $\palefunc$ of $\scalefunc$ 
satifies $\palefunc \in L^1(\Gamma, \bC)$.
There is a constant $\tK > 0 $, depending only on $\chil$, 
such that  if  $ \frac{2 \tK}{\Omega} \norm{F}_1< 1$ and $\Omega^{-1}\, \norm{E}_\infty  < 1$, 
the decay constant of $\Comhg$ satisfies
\begin{equation}\label{dComhg}
\abf_{\Comhg}
\le
\frac{\tK}{\Omega}\;\;
\frac{\norm{\palefunc}_1}{1- 2 \tK \Omega^{-1}\norm{F}_1} .
\end{equation}
In particular, if $\tK \norm{F}_1 < \frac14 \Omega$, then 
$\abf_{\Comhg} \le \frac{2\tK}{\Omega} \norm{\palefunc}_1$.
\end{lemma}

\begin{proof}
Let
\begin{equation}
u(\tau)  
=
\frac{1}{\beta}
\sum_{\omega \in \frac{\pi}{\beta} \bZ }
\E^{-\I \omega \tau} \; \chil \left(\frac{\omega}{\Omega}\right) 
\end{equation}
then 
\begin{equation}
\Comhg  = C^{(\scalefunc)} -  u * C^{(\scalefunc)}
\end{equation}
where the convolution is in $\tau$. 
By summation by parts, 
\begin{equation}
\left(
\E^{\I \frac{\pi}{\beta}\tau} -1
\right)^n
u (\tau) 
=
\frac{1}{\beta}
\sum_\omega
\E^{-\I \omega \tau}
(\delta^n \chil )
\left(\frac{\omega}{\Omega}\right)
\end{equation}
where $\delta$ is the difference operator 
$(\delta f) (\omega) = f(\omega+\frac{\pi}{\beta}) - f(\omega)$. 
Using that for all $\tau $ with $|\tau |  \le \beta$, 
$\abs{\E^{\I \frac{\pi}{\beta}\tau} -1} = 2 \sin \frac{\pi |\tau|}{2\beta} \ge
2 \frac{|\tau|}{\beta}$  and that $\chil$ is smooth,  it follows that 
\begin{equation}\label{96}
\abs{u(\tau)} 
\le 
\frac14 \tK \frac{\Omega}{(1+ \Omega |\tau|)^3}
\end{equation}
where $\tK$ depends on the sup norms of the first three derivatives of $\chil$.

Let $a(\tau) = \CC(\tau,0)$. 
By definition, $a(s) = \theta^+ (-s) - \frac12 $ where 
$\theta^+ (t) =1$ for $t \ge 1$ and zero otherwise. 
Because $\int_{-\beta}^\beta  u(s) \dd s = \chil (0) =1$, 
\begin{equation}
a(\tau) - (u*a)(\tau)
=
\int\limits_{-\beta}^\beta \dd s\; 
u(s) \;
[ a(\tau) - a(\tau -s)] .
\end{equation}
The $\frac12$ drops out, and  
$a(\tau) - (u*a)(\tau) = \mbox{ sgn}(\tau) \int_{\cI (\tau)} u(s) \dd s$, 
where
\begin{equation}
\cI (\tau)
=
\cases{
[-\beta,-\beta+\tau] \cup [\tau,\beta]  & for $\tau > 0$ \cr
& \cr
[-\beta,\tau] \cup [\beta+\tau, \beta] & for $\tau \le 0$.
}
\end{equation}
Our hypothesis on $\palefunc$ and (\ref{96}) imply that
\begin{equation}
\Aomhg (\tau, \Sp{x})
=
\palefunc (\Sp{x}) \; 
[ a(\tau) - (u*a)(\tau)]
\end{equation}
satisfies
\begin{equation}\label{onenorm}
\norm{\Aomhg (\tau, \Sp{x})}_1
\le
\tK
\norm{\palefunc}_1
\Omega^{-1} .
\end{equation}
The same bound holds with $(\palefunc,\scalefunc)$ replaced by $(F,E)$.
For all $(\tau,\Sp{x})$, 
\begin{equation}
\Comhg (\tau,\Sp{x})
=
\Aomhg (\tau,\Sp{x})
+
\frac{1}{\beta}
\sum_\omega
\frac{\E^{-\I \omega\tau}}{(\I \omega)^2}
\chig\left(\frac{\omega}{\Omega}\right)\;
\int_\cB \mumoma (\dd \Sp{p}) \;
\frac{E(\Sp{p}) \scalefunc (\Sp{p})}{ 1- \frac{E(\Sp{p})}{\I \omega} }\;
\E^{\I \Sp{p} \Sp{x}} .
\end{equation}
Because $\chig(\frac{\omega}{\Omega}) = 0$ for 
 $|\omega| \le \Omega$, the condition $\Omega^{-1} \norm{E}_\infty < 1$
implies that the geometric series for 
$(1-E(\Sp{p})/\I \omega )^{-1}$ converges uniformly in $\Sp{p}$. 
By dominated convergence, the summation can be exchanged with the
integral over $\Sp{p}$ and the summation over $\omega$. 
Moreover, by the support properties of $\chil$, we may 
insert a factor $\chig (\frac{2\omega}{\Omega})^n$ 
in the $n^{\rm th}$ order term in this expansion, to get
\begin{equation}
\Comhg (\tau,\Sp{x})
=
\Aomhg (\tau,\Sp{x})
+
\left[
\sum_{n=1}^\infty
\Aomhg * 
A_{\Omega/2}^{(E,>)}
* \ldots * 
A_{\Omega/2}^{(E,>)}
\right] (\tau, \Sp{x})
\end{equation}
where the convolution is in $\tau $ and $\Sp{x}$ and 
$n$ factors $A_{\Omega/2}^{(E,>)}$ appear in the  product. 
The standard $L^1$ bound for the convolution implies
\begin{equation}
\norm{\Comhg}_1
\le
\norm{\Aomhg}_1
\left(
1
+
\sum_{n=1}^\infty
\norm{A_{\Omega/2}^{(E,>)}}_1^n
\right)
\end{equation}
which converges  by the hypotheses on $\palefunc$, $F$, and 
by (\ref{onenorm}), and yields the bound (\ref{dComhg}).
\end{proof}

\bigskip\noindent
Theorem \ref{konv} directly applies and implies convergence of 
the effective action obtained from the integration over large frequencies. 
Note that because of the way the constants depend on $\Omega$, 
the initial interaction can be taken arbitrarily strong (as long as it is 
summable):
if $U$ denotes the coupling constant of a quartic interaction, 
convergence of the expansion for the effective action holds for all $U$ with 
$\frac{U}{\Omega} (\ln \Omega)^2 $ small enough, which can always
be achieved by taking $\Omega$ large enough. Thus, for arbitrarily 
strong coupling, the initial integration step is given by a convergent 
expansion. The consequences and some possible extensions of this 
are discussed in \cite{dynRG}.

\bigskip\noindent
{\bf Acknowledgement. } This work was supported by 
DFG grant Sa~1362/1, by the Max--Planck society, and, 
in its final stage, by 
{\em DFG-Forschergruppe} FOR718.


\begin{thebibliography}{99}

\bibitem[$A\Gamma\Delta$]{AGD}
A.A.\ Abrikosov, L.P.\ Gorkov, I.E.\ Dzyaloshinski, 
{\sl Methods of Quantum Field Theory in Statistical Physics},
Dover, 1963

\bibitem[AR]{AbdessRiv98} 
A.~Abdesselam, V.~Rivasseau, {\sl Explicit fermionic tree
expansions}, Lett.\ Math.\ Phys.\  {\bf 44} 77--88 (1998).

\bibitem[BR]{BratRob} 
O.\ Bratteli, D.W.\ Robinson, 
{\sl Operator algebras and quantum statistical mechanics}, 
vol. 1,2; Springer

\bibitem[FKT98]{FKT98}
J. Feldman, H. Kn{\"o}rrer, E. Trubowitz. {\sl A Representation for
Fermionic Correlation Functions}. Commun. Math. Phys. {\bf 195},
465--493 (1998).

\bibitem[FKT02]{FKTbook}
J. Feldman, H. Kn{\"o}rrer, E. Trubowitz. 
{\sl Fermionic Functional Integrals and the Renormalization Group},
CRM Monograph series, vol.\ 16, AMS, 2002

\bibitem[FKT04]{FKT04}
J. ~Feldman, H. ~Kn{\"o}rrer, E. ~Trubowitz, 
{\sl Convergence of
Pertubation Expansions in Fermionic Models: Part 1},
 Commun. \ Math. \ Phys. \ {\bf 247} 195--242 (2004).

\bibitem[S98a]{crg}
M.\ Salmhofer, 
{\sl Continuous renormalization for fermions and Fermi liquid theory},
 Commun. \ Math. \ Phys. \ {\bf  194} (1998) 249--295

\bibitem[S98b]{msbook} 
M.\ Salmhofer, {\sl Renormalization}, Springer, Heidelberg, 1998

\bibitem[S07]{dynRG}
M.~Salmhofer, 
{\sl Dynamical adjustment of propagators in renormalization group flows}, 
Ann.\ Phys.\ (Leipzig) {\bf 16}, No. 3 (2007) 171--206

\bibitem[SW]{Salm99} 
M.~Salmhofer, C.~Wieczerkowski, {\sl Positivity and convergence in
fermionic quantum field theory}, J.\ Stat.\ Phys.\  {\bf 99},  557--586 (2000).

\bibitem[PS]{PeSa}
W.\ Pedra, M.\ Salmhofer, {\sl On the Mathematical Theory of Fermi Liquids in 
Two Dimensions}, to appear

\bibitem[St]{Stein}
E.\ Stein, {\em Harmonic Analysis}, Princeton University Press, 1993, Chapter 3

\end{thebibliography}
\end{document}